\documentclass[11pt]{article}
\usepackage[utf8]{inputenc}
\usepackage{amsmath}
\usepackage{authblk}

\usepackage[margin=1in]{geometry}
\usepackage{amsmath,amsthm,amssymb,xspace,verbatim,paralist,enumitem,multirow,bm,bbm}
\usepackage{array,multirow,multicol,booktabs}
\usepackage{graphicx}
\usepackage{algorithm,algorithmicx}
\usepackage[noend]{algpseudocode}

\usepackage{url}
\usepackage[colorlinks, linktoc=all, linkcolor=BrickRed, citecolor=blue]{hyperref}

\usepackage{esvect}

\usepackage[nameinlink]{cleveref}

\usepackage[usenames, dvipsnames]{xcolor}

\usepackage{mathdots,mathtools,mathrsfs}
\usepackage{tikz}

\usepackage{subcaption}

\definecolor{airforceblue}{rgb}{0.36, 0.54, 0.66}
\definecolor{bluegray}{rgb}{0.4, 0.6, 0.8}
\definecolor{ceil}{rgb}{0.57, 0.63, 0.81}
\definecolor{celestialblue}{rgb}{0.29, 0.59, 0.82}
\definecolor{cerulean}{rgb}{0.0, 0.48, 0.65}
\definecolor{celadon}{rgb}{0.67, 0.88, 0.69}
\definecolor{yellow-green}{rgb}{0.6, 0.8, 0.2}

\colorlet{graph-vertex-blue}{blue}
\colorlet{graph-vertex-green}{Green}
\colorlet{graph-vertex-red}{red}
\colorlet{graph-vertex-yellow}{yellow}

\colorlet{graph-edge-blue}{blue}
\colorlet{graph-edge-green}{Green}
\colorlet{graph-edge-red}{red}
\colorlet{graph-edge-yellow}{yellow}

\usepackage{thmtools}
\usepackage{thm-restate}

\newtheorem{theorem}{Theorem}[section]
\newtheorem{lemma}[theorem]{Lemma}
\newtheorem{claim}[theorem]{Claim}

\newtheorem{conjecture}[theorem]{Conjecture}

\newtheorem{observation}[theorem]{Observation}
\theoremstyle{definition}  \newtheorem{definition}[theorem]{Definition}

\theoremstyle{remark}

\usepackage[normalem]{ulem} 

\newbox\mybox 
\newdimen\myboxwidth    

\newcommand\addpicture[3]{%
\setbox\mybox=\hbox{\includegraphics[scale=#3]{#2}}
\myboxwidth\wd\mybox    
\renewcommand\windowpagestuff{%
\includegraphics[scale=#3]{#2}
\captionof{figure}{A test figure.}}
\parpic[#1]{%
\begin{minipage}{\myboxwidth}
 \windowpagestuff 
\end{minipage} 
} }

\newcommand{\ignore}[1]{}

\newcommand{\mypar}[1]{\medskip\noindent{\sffamily\bfseries #1.}~}
\newcommand{\etal}{{et al.}\xspace}

\newcommand{\cF}{\cal F}

\newcommand{\cP}{\cal P}

\newcommand{\cT}{\cal T}

\newcommand{\mcB}{\mathcal B}
\newcommand{\mcC}{\mathcal C}

\newcommand{\mcE}{\mathcal E}

\newcommand{\mcP}{\mathcal P}

\newcommand{\mcS}{\mathcal S}
\newcommand{\mcT}{\mathcal T}

\newcommand{\NN}{\mathbb{N}}

\usepackage{hyperref}
\newcommand{\Sec}[1]{\hyperref[sec:#1]{\S\ref*{sec:#1}}} 
\newcommand{\Eqn}[1]{\hyperref[eq:#1]{(\ref*{eq:#1})}} 
\newcommand{\Fig}[1]{\hyperref[fig:#1]{Fig.\,\ref*{fig:#1}}} 
\newcommand{\Tab}[1]{\hyperref[tab:#1]{Tab.\,\ref*{tab:#1}}} 
\newcommand{\Thm}[1]{\hyperref[thm:#1]{Theorem\,\ref*{thm:#1}}} 
\newcommand{\Fact}[1]{\hyperref[fact:#1]{Fact\,\ref*{fact:#1}}} 
\newcommand{\Lem}[1]{\hyperref[lem:#1]{Lemma\,\ref*{lem:#1}}} 
\newcommand{\Prop}[1]{\hyperref[prop:#1]{Prop.~\ref*{prop:#1}}} 
\newcommand{\Cor}[1]{\hyperref[cor:#1]{Corollary~\ref*{cor:#1}}} 
\newcommand{\Conj}[1]{\hyperref[conj:#1]{Conjecture~\ref*{conj:#1}}} 
\newcommand{\Def}[1]{\hyperref[def:#1]{Definition~\ref*{def:#1}}} 
\newcommand{\Alg}[1]{\hyperref[alg:#1]{Alg.~\ref*{alg:#1}}} 
\newcommand{\Ex}[1]{\hyperref[ex:#1]{Ex.~\ref*{ex:#1}}} 
\newcommand{\Clm}[1]{\hyperref[clm:#1]{Claim~\ref*{clm:#1}}} 

\renewcommand{\geq}{\geqslant}
\renewcommand{\leq}{\leqslant}

\DeclareMathOperator*{\argmin}{arg\,min}

\newcommand{\dir}{^\rightarrow}

\newcommand{\degen}{\kappa}

\newcommand{\TRICONJ}{\textsc{Triangle Detection Conjecture}\xspace}
\newcommand{\SUBC}{\mathrm{\#Sub}\xspace}
\newcommand{\HOMC}{\mathrm{\#Hom}}
\newcommand{\TRIC}{\textsc{tri-cnt}\xspace}

\newcommand{\dtd}{{DAG tree decomposition}\xspace}
\newcommand{\pdtd}{{partial DAG tree decomposition}\xspace}
\newcommand{\dtwtext}{{DAG treewidth}\xspace}
\newcommand{\pdtwtext}{{partial DAG treewidth}\xspace}
\newcommand{\good}{\textsc{good-pair}\xspace}

\DeclareMathOperator{\dtw}{\tau}
\DeclareMathOperator{\tw}{tw}
\DeclareMathOperator{\uniqueH}{UR}

\newcommand{\fixed}{\textsc{fixed-set}\xspace}
\newcommand{\core}{\textsc{core-set}\xspace}
\newcommand{\aux}{\textsc{auxiliary-set}\xspace}

\newcommand{\excludeC}{C\textsc{-excluded}\xspace}

\newcommand{\coreGraph}{\textsc{core}\xspace}
\newcommand{\fixedGraph}{\textsc{fixed}\xspace}
\newcommand{\connect}{\textsc{bridge}\xspace}

\DeclareMathOperator{\pmatch}{{\cP}-match}
\DeclareMathOperator{\ppmatch}{{\cP^\prime}-match}
\DeclareMathOperator{\phom}{{\cP}-homomorphism}
\DeclareMathOperator{\map}{\sigma}

\newcommand{\triMap}{\textsc{triangle}\xspace}

\DeclareMathOperator{\Hmatch}{H-match}
\DeclareMathOperator{\ALG}{ALG}

\DeclareMathOperator{\dist}{dist}

\newcommand{\isc}{\textsc{intersection-cover}\xspace}
\newcommand{\cover}[1]{{#1}\textsc{-cover}\xspace}

\newcommand{\R}{\textsc{reachable}\xspace}
\newcommand{\UR}{\textsc{unique-reachable}\xspace}
\newcommand{\Hdir}{H\dir}

\newcommand{\Hom}[2]{\mathrm{Hom}_#2(#1)}

\DeclareMathOperator{\LICS}{LICL}

\renewcommand{\mcS}{S}

\usepackage{tikz}
\usetikzlibrary{decorations.markings,fit,arrows,positioning,backgrounds}
\tikzset{%
  gnode/.style={shape=circle,minimum size=3mm,fill,draw=black}
}
\tikzset{myptr/.style={decoration={markings,mark=at position 1 with {\arrow[scale=1.5,>=stealth]{>}}},postaction={decorate}}}

\tikzset{myptr2/.style={decoration={markings,mark=at position 0.55 with {\arrow[scale=1.5,>=stealth]{>}}},postaction={decorate}}}

\def\URCycle{

   \node (9)  [label={below:\small $s_3$}] at (7.9,0) [nd, fill =] {};
   \node (10) [label={left: \small $v_{3,1}$}] at (7,0.75) [nd, fill =] {};
   \node (11) [label={left: \small $s_1$}]  at (7,1.5) [nd] {};
   \node (12) [label={above:\small $v_{1,2}$}] at (7.9,2.25) [nd] {};
   \node (13) [label={right:\small $s_2$}] at (8.8,1.5) [nd] {};
   \node (14) [label={right:\small $v_{2,3}$}] at (8.8,0.75) [nd] {};
   \draw[myptr] (9) to (10);
   \draw[myptr] (9) to (14);
   \draw[myptr] (11) to (10);
   \draw[myptr] (11) to (12);
   \draw[myptr] (13) to (12);
   \draw[myptr] (13) to (14);

   \node [inner sep = 0,below] at +(8,-0.7) {{\Large $C^\prime$}};

   \node (15)  [label={below:\small $s_3$}] at (1.9,0) [nd, fill =] {};
   \node (16) [label={left: \small $s_1$}]  at (1,1.5) [nd] {};
   \node (17) [label={right:\small $s_2$}] at (2.8,1.5) [nd] {};
   \draw (15) to (16);
   \draw (15) to (17);
   \draw (16) to (17);
   
   \node [inner sep = 0,below] at +(2,-0.7) {{\Large $\uniqueH_{\mcS_p}$}};
}

\def\DTDSixCycle{

   \node (1) at (0.9,0) [nd] {};
   \node (2) at (0,0.75) [nd] {};
   \node (3) at (0,1.5) [nd] {};
   \node (4) at (0.9,2.25) [nd] {};
   \node (5) at (1.8,1.5) [nd] {};
   \node (6) at (1.8,0.75) [nd] {};
   \draw (1) to (2);
   \draw (1) to (6);
   \draw (3) to (2);
   \draw (3) to (4);
   \draw (5) to (4);
   \draw (5) to (6);
   
   \node [inner sep = 0,below] at +(1,-0.7) {{\Large $H$}};

   \node (9)  [label={below:\small $s_3$}] at (4.9,0) [nd, fill =] {};
   \node (10) [label={left: \small $t_2$}] at (4,0.75) [nd, fill =] {};
   \node (11) [label={left: \small $s_1$}]  at (4,1.5) [nd] {};
   \node (12) [label={above:\small $t_3$}] at (4.9,2.25) [nd] {};
   \node (13) [label={right:\small $s_2$}] at (5.8,1.5) [nd] {};
   \node (14) [label={right:\small $t_1$}] at (5.8,0.75) [nd] {};
   \draw[myptr] (9) to (10);
   \draw[myptr] (9) to (14);
   \draw[myptr] (11) to (10);
   \draw[myptr] (11) to (12);
   \draw[myptr] (13) to (12);
   \draw[myptr] (13) to (14);

   \node [inner sep = 0,below] at +(5,-0.7) {{\Large $H\dir$}};

   \node (15)  [label={below:\small $s_3$}] at (8.9,0) [nd, fill =] {};
   \node (16) [label={left: \small $s_1$}]  at (8,1.5) [nd] {};
   \node (17) [label={right:\small $s_2$}] at (9.8,1.5) [nd] {};
   \draw (15) to (16);
   \draw (15) to (17);
   \draw (16) to (17);
   
   \node [inner sep = 0,below] at +(9,-0.7) {{\Large $\uniqueH_{\mcS(\Hdir)}$}};
}

\def\exampleH{
   \node (1) [label={above: $a_1$}] at (1.6,4.5) [nd] {};
   \node (2) [label={above: $a_2$}] at (2.4,4.5) [nd] {};
   \node (3) [label={below: $a_3$}] at (0.9,2) [nd] {};
   \node (4) [label={left:  $a_4$}] at (0,1) [nd] {};
   \node (5) [label={below: $a_5$}] at (1.3,0) [nd] {};
   \node (6) [label={below: $a_6$}] at (2.7,0) [nd] {};
   \node (7) [label={right: $a_7$}] at (4.25,1) [nd] {};
   \node (8) [label={below: $a_8$}] at (3.1,2) [nd] {};

   \draw (1) to (2);
   \draw (1) to (3);
   \draw (1) to (4);
   \draw (1) to (5);
   \draw (2) to (6);
   \draw (2) to (7);
   \draw (2) to (8);
   \draw (3) to (4);
   \draw (3) to (8);
   \draw (4) to (5);
   \draw (5) to (6);
   \draw (6) to (7);
   \draw (7) to (8);

   \draw [very thick,rotate around={90:(2,4.5)},red!85!black] (2,4.5) ellipse (20pt and 50pt);
   \node [red!85!black] at (2,5.75) {\Large $H_{\excludeC}$};
}

\def\GHexampleH{
   \node (1) [label={above: \Large $z_1$}] at (3.2,6) [nd,fill=NavyBlue] {};
   \node (2) [label={above: \Large $z_2$}] at (4.8,6) [nd,fill=NavyBlue] {};
   \node (3) [label={below: \Large $V_1$}] at (1.8,2) [nd] {};
   \node (4) [label={[label distance=24pt] left: \Large $V_{1,2}$}] at (0,1) [nd] {};
   \node (5) [label={below: \Large $V_2$}] at (2.6,0) [nd] {};
   \node (6) [label={[label distance=14pt] below: \Large $V_{2,3}$}] at (5.4,0) [nd] {};
   \node (7) [label={[label distance=-1.5pt] right: \Large $V_3$}] at (8.5,1) [nd] {};
   \node (8) [label={[label distance=14pt] below: \Large $V_{1,3}$}] at (6.2,2) [nd] {};
   
   \draw (1) to (2);
   \draw [thick, violet] (1) to (3);
   \draw [thick, violet] (1) to (4);
   \draw [thick, violet] (1) to (5);
   \draw [thick, violet] (2) to (6);
   \draw [thick, violet] (2) to (7);
   \draw [thick, violet] (2) to (8);
   \draw (3) to (4);
   \draw (3) to (8);
   \draw (4) to (5);
   \draw (5) to (6);
   \draw (6) to (7);
   \draw (7) to (8);
   
   \draw [thick, dashed, green!75!black] (1.8,2) circle (24pt);
   
    \draw [very thick, dotted, red, rotate around={90:(0,1)}] (0,1) ellipse (15pt and 25pt);
   
   \draw [thick, dashed, green!75!black] (2.6,0) circle (24pt);
   
   \draw [very thick, dotted, red, rotate around={-90:(5.4,0)}] (5.4,0) ellipse (15pt and 25pt);
   
   \draw [thick, dashed, green!75!black] (8.5,1) circle (24pt);
   
   \draw [very thick, dotted, red, rotate around={-90:(6.2,2)}] (6.2,2) ellipse (15pt and 25pt);

   \draw [very thick,rotate around={90:(4,1)},blue!85!black] (4,1) ellipse (80pt and 175pt);
   \node [blue!85!black] at (9,3.75) {\huge $G_{\coreGraph}$};
   
   \draw [very thick,rotate around={90:(4,6)},red!85!black] (4,6) ellipse (25pt and 60pt);
   \node [red!85!black] at (7.5,6) {\huge $G_{\fixedGraph}$};
   
   \node [violet] at (0.5,4.75) {\huge $E_{\connect}$};
}

\newcommand*\samethanks[1][\value{footnote}]{\footnotemark[#1]}

\begin{document}


\title{Near-Linear Time Homomorphism Counting in Bounded Degeneracy Graphs: The Barrier of Long Induced Cycles}


\author{Suman K. Bera\thanks{All the authors are supported by NSF TRIPODS grant CCF-1740850, NSF CCF-1813165, CCF-1909790, and ARO Award W911NF1910294.}}
\author{Noujan Pashanasangi\samethanks}
\author{C. Seshadhri\samethanks}
\affil{University of California, Santa Cruz}
\affil{\{sbera,npashana,sesh\}@ucsc.edu}
\date{}

\maketitle

\begin{abstract}
Counting homomorphisms of a constant sized pattern graph $H$ in an input graph $G$
is a fundamental computational problem. There is a rich history of studying the complexity of this problem, under various constraints on the input $G$ and the pattern $H$. Given the significance of this problem and the large sizes of modern inputs, we investigate when near-linear time algorithms are possible. We focus on the case when the input graph has bounded degeneracy, a commonly studied and practically relevant class for homomorphism counting. It is known from previous work that for certain classes of $H$, $H$-homomorphisms can be counted exactly in near-linear time in bounded degeneracy graphs. Can we precisely characterize the patterns $H$ for which near-linear time algorithms are possible?

We completely resolve this problem, discovering a clean dichotomy using fine-grained complexity. Let $m$ denote the number of edges in $G$. We prove the following: if the largest induced cycle in $H$ has length at most $5$, then there is an $O(m\log m)$ algorithm for counting $H$-homomorphisms in bounded degeneracy graphs. If the largest induced cycle in $H$ has length at least $6$, then (assuming standard fine-grained complexity conjectures) there is a constant $\gamma > 0$, such that there is no $o(m^{1+\gamma})$ time algorithm for counting $H$-homomorphisms.


\end{abstract}

\section{Introduction}

Analyzing occurrences of small pattern graphs in a large input graph is a fundamental algorithmic problem with a rich line of work in both theory and practice~\cite{lovasz1967operations,chiba1985arboricity,flum2004parameterized,dalmau2004complexity,lovasz2012large,AhNe+15,curticapean2017homomorphisms,PiSeVi17,roth2020counting}.
A common version of this problem is \emph{homomorphism counting}, which has numerous applications in logic, properties of graph products, partition functions in statistical physics, database theory, and network science~\cite{chandra1977optimal,brightwell1999graph,dyer2000complexity,borgs2006counting,PiSeVi17,dell2019counting,pashanasangi2020efficiently}. Denote the pattern simple graph 
as $H = (V(H), E(H))$, which is thought of as fixed (or constant-sized). The input simple graph is denoted by $G = (V(G), E(G))$. An $H$-homomorphism is an edge-preserving map $f:V(H) \to V(G)$. Formally, $\forall (u,v) \in E(H)$, $(f(u), f(v)) \in E(G)$. When $f$ is an injection, this map corresponds to the common notion of subgraphs. We use $\Hom{G}{H}$ to denote the count of the distinct $H$-homomorphisms. 

The problem of computing $\Hom{G}{H}$ for various choices of $H$ is a deep subfield of study in graph algorithms~\cite{itai1978finding,alon1997finding,brightwell1999graph,dyer2000complexity,diaz2002counting,dalmau2004complexity,borgs2006counting,curticapean2017homomorphisms,bressan2019faster,roth2020counting}. Indeed, even the case of $H$ being a triangle, clique, or cycle have led to long lines of results. Many practical and theoretical algorithms for subgraph counting are based on homomorphism counting, and it is known
that subgraph counts can be expressed as a linear combination of homomorphism counts~\cite{curticapean2017homomorphisms}.
Let $n = |V(G)|$ and $k = |V(H)|$. 
It is known that computing $\Hom{G}{H}$ is $\#W[1]$-hard when parameterized by $k$ (even when $H$ is a $k$-clique), so
we do not expect $n^{o(k)}$ algorithms for general $H$~\cite{dalmau2004complexity}. Yet the $n^k$
barrier can be beaten when $H$ has  structure. Notably, Curticapean-Dell-Marx proved that if $H$ has treewidth at most 2, then $\Hom{G}{H}$ can be computed in $\text{poly}(k) \cdot n^\omega$, where $\omega$ is the matrix multiplication constant~\cite{curticapean2017homomorphisms}. 

For many modern applications in network science, since $n$ is very large, one desires
linear-time algorithms. Indeed, one might argue that the ``right'' notion of computational efficiency in these settings is (near) linear time. Motivated by these concerns, we investigate the barriers for achieving linear time algorithms to count $\Hom{G}{H}$, especially when $G$ is a ``real-world'' graph.

We focus on the class of \emph{bounded degeneracy graphs}. This is the class of graphs where all subgraphs have a constant average degree. A seminal result of Chiba-Nishizeki proves that clique counting can be done in linear time for such graphs~\cite{chiba1985arboricity}. Importantly, this paper introduced the technique of
graph orientations for subgraph counting, that has been at the center of many 
state-of-the-art practical  algorithms~\cite{AhNe+15,jha2015path,PiSeVi17,ortmann2017efficient,jain2017fast,pashanasangi2020efficiently}. The degeneracy has a special significance in the analysis of real-world graphs, since it is intimately tied to the technique of ``core decompositions''. 

Many algorithmic ideas for homomorphism or subgraph counting on bounded degeneracy graphs have been quite successful in practice, and form the basis
of state-of-the-art algorithms.
Can we understand the limits of these methods? Theoretically, when is homomorphism counting possible in near-linear time on bounded degeneracy graphs? Chiba-Nishizeki~\cite{chiba1985arboricity} proved that clique and 4-cycle counting can be done in linear time. A result of the authors shows that near-linear time is possible when $|V(H)| = k \leq 5$~\cite{bera2020linear}. In  a significant generalization, Bressan~\cite{bressan2019faster} defines an intricate notion of \emph{DAG treewidth}, and shows (among other things) that a near-linear time algorithm exists when the DAG treewidth of $H$ is one.
These results lead us to the main question addressed by our work.

\medskip

\emph{Can we characterize the pattern graphs $H$ for which $\Hom{G}{H}$ is computable
in near-linear time (when $G$ has bounded degeneracy)?}

\medskip

Our main result is a surprisingly clean resolution of this problem, assuming
fine-grained complexity results. There is a rich history of {\em complexity dichotomies} for homomorphism detection and counting problems~\cite{hell1990complexity,dyer2000complexity,dalmau2004complexity,grohe2007complexity,roth2020counting}.
In this work, we discover such a dichotomy for near-linear time algorithms.

Let $\LICS(H)$ be the length of the largest induced cycle in $H$. 

\begin{theorem} \label{thm:main} Let $G$ be an input graph with $n$ vertices, $m$
edges, and degeneracy $\degen$. Let $f:\NN \to \NN$ denote some explicit function. Let $\gamma > 0$ denote the constant from the \TRICONJ.

If $\LICS(H) \leq 5$: there exists an algorithm that computes $\Hom{G}{H}$ in
time $f(\degen)\cdot m\log n$. 

If $\LICS(H) \geq 6$: assume the \TRICONJ. For any function $g: \NN \to \NN$, there is no 
algorithm with (expected) running time
$g(\degen) o(m^{1+\gamma})$
that computes $\Hom{G}{H}$.
\end{theorem}

(We note that the condition on $H$ involves induced cycles, but we are interested
in counting non-induced homomorphisms.)

Abboud and Williams introduced the \TRICONJ on the complexity of determining whether a graph has a triangle~\cite{abboud2014popular}. Assuming that there is no $O(m^{1+\gamma})$ time  triangle detection algorithm, they proved lower bounds for many classic graph algorithm problems. It is believed that the
constant $\gamma$ could be arbitrarily close to $1/3$~\cite{abboud2014popular}.

\begin{conjecture}[\TRICONJ~\cite{abboud2014popular}]
\label{conj:triangle}
There exists a constant $\gamma>0$ such that in the word RAM model
of $O(\log n)$ bits, any algorithm to detect whether an input graph on
$m$ edges has a triangle requires $\Omega(m^{1+\gamma})$ time
in expectation.
\end{conjecture}

\subsection{Main Ideas} \label{sec:ideas}

\mypar{Background for the Upper Bound} We begin with some context on the
main algorithmic ideas used for homomorphism/subgraph counting in bounded degeneracy graphs. Any graph $G$ of bounded degeneracy has an acyclic orientation $G^\rightarrow$, where all outdegrees are bounded. Moreover,  $G^\rightarrow$ can be found in linear time~\cite{matula1983smallest}. For any pattern graph $H$, we consider all possible acyclic orientations. For each such orientation $H^\rightarrow$, we compute
the number of $H^\rightarrow$-homomorphisms (in $G^\rightarrow$). (Directed homomorphisms are maps that preserve the direction of edges.) Finally,
we sum these counts over all acyclic orientations $H^\rightarrow$.
This core idea was embedded
in the seminal paper of Chiba-Nishizeki, and has been presented in such terms
in many recent works~\cite{PiSeVi17,ortmann2017efficient,bressan2019faster,bera2020linear}.

Since $G^\rightarrow$ has bounded outdegrees, for any bounded,
rooted tree $T^\rightarrow$ (edges pointing towards leaves), \emph{all} 
$T^\rightarrow$-homomorphisms can be explicitly enumerated in linear time. To construct
a homomorphism of $H^\rightarrow$, consider the rooted trees 
of a DFS forest $T^\rightarrow_1$, $T^\rightarrow_2, \ldots$ generated
by processing the sources first. We first enumerate all homomorphisms of
$T^\rightarrow_1$, $T^\rightarrow_2, \ldots$ in linear time. We need to count how
many tuples of these homomorphisms
can be ``assembled'' into $H^\rightarrow$-homomorphisms. (We note that the number
of $H^\rightarrow$-homomorphisms can be significantly super-linear.) The main 
idea is to index the rooted tree homomorphisms appropriately, so that 
$H^\rightarrow$-homomorphisms can be counted in linear time. This requires a
careful understanding of the shared vertices among the rooted DFS forest
$T^\rightarrow_1$, $T^\rightarrow_2, \ldots$.

The previous work of the authors showed how this efficient counting
can be done when $|V(H)| \leq 5$, though the proof was  ad hoc~\cite{bera2020linear}. It did
a somewhat tedious case analysis for various $H$, exploiting specific structure
in the various small pattern graphs. Bressan gave a remarkably principled approach, introducing
the notion of the \emph{DAG treewidth}~\cite{bressan2019faster}. We will take some liberties with
the original definition, for the sake of exposition. Bressan defined the DAG
treewidth of $H^\rightarrow$, and showed that when this quantity is $1$, 
$\Hom{G^\rightarrow}{H^\rightarrow}$ can be computed in near-linear time.
The DAG treewidth is $1$ when the following construct exists. For any source
$s$ of $H^\rightarrow$, let $R(s)$ be the set of vertices in $H^\rightarrow$ reachable
from $s$. The sources of $H^\rightarrow$ need to be be arranged in a tree $\cT$ such
that the following holds. If $s$ lies on the (unique) path between $s_1$ and $s_2$ (in $\cT$), then $R(s_1) \cap R(s_2) \subseteq R(s)$. In some sense, this gives
a divide-and-conquer framework to construct (and count) $H^\rightarrow$-homomorphisms. 
Any $H^\rightarrow$-homomorphism can be broken into ``independent pieces''
that are only connected by the restriction of the homomorphism to $R(s)$. By
indexing all the tree homomorphisms appropriately, the total count of $H^\rightarrow$-homomorphisms can be determined in near-linear time by dynamic programming.
Note that we need the DAG treewidth of \emph{all} acyclic orientations of $H$ to be $1$, which is a challenging notion to describe succinctly.

\mypar{From Induced Cycles to DAG tree decompositions} We observe an interesting
contrast between the previous work of the authors and Bressan's work. The former provides a simple family of $H$ for which $\Hom{G}{H}$ can be computed
in near-linear time in bounded degeneracy graphs, yet the proofs were ad hoc. The latter gave a principled
algorithmic approach, but it does not succinctly describe what kinds of $H$ allow
for such near-linear algorithms. Can we get the best of both worlds?

Indeed, that is what we achieve. By a deeper understanding of \emph{why} $|V(H)| \leq 5$ was critical in~\cite{bera2020linear} and generalizing it through the language of 
DAG tree decompositions, we can prove: the DAG treewidth of $H$ is one iff $\LICS(H) \leq 5$. 

When $\LICS(H) \leq 5$, for any acyclic orientation $H^\rightarrow$, we provide a (rather complex) iterative procedure to construct
the desired DAG tree decomposition $\cT$. The proof is intricate and involves many
moving parts. The connection between induced cycles
and DAG tree decompositions is provided by a construct called the \emph{unique reachability graph}. For any set $S$ of sources in $H^\rightarrow$, construct
the following simple, undirected graph $UR(S)$. Add edge $(s,s')$ if there
exists a vertex that is in $R(s) \cap R(s')$, but not contained in any $R(s'')$, for
$s'' \in S \setminus \{s,s'\}$. A key lemma states that if $UR(S)$ contains a cycle
(for any subset $S$ of sources), then $H$ contains an induced cycle of at least twice the length. Any cycle in a simple graph has length at least $3$. So if $UR(S)$ has
a cycle, then $H$ has an induced cycle of length at least $6$. Thus, if $\LICS(H) \leq 5$, for all $S$, the simple graph $UR(S)$ is a forest.

For any set $S$ of sources, we will (inductively) construct a \emph{partial} DAG tree decomposition that only involves $S$. Let us try to identify a ``convenient'' vertex $x \in S$ with the following property. We inductively take the partial DAG tree decomposition $\cT'$ of $S \setminus \{x\}$, and try to attach $x$ as a leaf in $\cT'$
preserving the DAG tree decomposition conditions (that involve reachability). By carefully working out the definitions, we identify a specific intersection
property of $R(x)$ with the reachable sets of the other sources in $S \setminus \{x\}$. When this property holds, we can attach $x$ and extend the partial DAG tree decomposition, as described above. When the property fails, we prove that the degree
of $x$ in $UR(S)$ is at least $2$. But $UR(S)$ is a forest, and thus contains a vertex
of degree $1$. Hence, we can always identify a convenient vertex $x$, and can iteratively build the entire DAG tree decomposition. 

We also prove the converse. If $\LICS(H) \geq 6$, then the DAG treewidth (of some
orientation) is at least two. This proof is significantly less complex, but
crucially uses the unique reachability graph.

\mypar{The Lower Bound: Triangles Become Long Induced Cycles} We start with the
simple construction of ~\cite{bera2020linear} that reduces triangle counting
in arbitrary graphs to $6$-cycle counting in bounded degeneracy graphs.
Given a graph $G$ where we wish to count triangles, we consider the graph
$G'$ where each edge of $G$ is subdivided into a path of length $2$. Clearly,
triangles in $G$ have a 1-1 correspondence with $6$-cycles in $G'$. It is
easy to verify that $G'$ has bounded degeneracy.

Our main idea is to generalize this idea for any $H$ where $\LICS(H) = 6$.
The overall aim is to construct a graph $G'$ where each $H$-homomorphism corresponds
to a distinct induced $6$-cycle in $G'$, which comes from a triangle in $G$. We will
actually fail to achieve this aim, but get ``close enough'' to prove the lower bound.

Let $\overline{H}$ denote the pattern obtained after removing the induced $6$-cycle
from $H$. Let us outline the construction of $G'$. We first take three copies
of the vertices of $G$. For every edge $(u,v)$ of $G$, connect copies of $u$
and $v$ that lie in \emph{different} copies by a path of length two. Note that
each triangle of $G$ has been converted into six $6$-cycles. We then add a single
copy of $\overline{H}$, and connect $\overline{H}$ to the remaining vertices
(these connections depend on the edges of $H$). This completes the description of $G'$.
Exploiting the relation of degeneracy to vertex removal orderings, 
we can prove that $G'$ has 
bounded degeneracy.

It is easy to see that every triangle in $G$ leads to a distinct $H$-homomorphism.
Yet the converse is potentially false. We may have ``spurious" $H$-homomorphisms that do
not involve the induced $6$-cycles that came from triangles in $G$.
By a careful analysis of $G'$, we can show the following. Every spurious
$H$-homomorphism \emph{avoids} some vertex in the copy of $\overline{H}$ (in $G'$).

These observations motivate the problem of \emph{partitioned}-homomorphisms.
Let $\mathbf{P}$ be a partition of the vertices of $G'$ into $k$ sets.
A partitioned-homomorphism is an $H$-homomorphism where each vertex
is mapped to a different set of the partition. We can choose $\mathbf{P}$ 
appropriately, so that the triangle count of $G$ is the number
of partitioned-homomorphisms scaled by a constant (that only depends
on the automorphism group of $H$). Thus, we reduce triangle counting
in arbitrary graphs to counting partitioned-homomorphisms in bounded degeneracy graphs.

Our next insight is to give up the hope of showing a many-one linear-time reduction
from triangle counting to $H$-homomorphisms, and instead settle
for a Turing reduction. This suffices for the lower bound of \Thm{main}. Using
inclusion-exclusion, we can reduce a single instance of partitioned-homomorphism
counting to $2^k$ instances of vanilla $H$-homomorphism counting. The details
are somewhat complex, but this description covers the basic ideas.

When $\LICS(H) > 6$, we replace edges in $G$ by longer paths, to give longer
induced cycles. The partitions become more involved, but the essence of the proof
remains the same.

\section{Related Work}
\label{sec:related}
Counting homomorphisms has a rich history in the
field of parameterized complexity theory.
D{\'\i}az~\etal~\cite{diaz2002counting} designed a 
dynamic programming based algorithm for the $\Hom{G}{H}$
problem with runtime $O(2^{k}n^{\tw(H)+1})$ where $\tw(H)$
is the treewidth of the target graph $H$.
Dalmau and Jonsson~\cite{dalmau2004complexity} proved
that $\Hom{G}{H}$ is polynomial time solvable if and
only if $H$ has bounded treewidth, otherwise it is
$\#W[1]$-complete. More recently, Roth and Wellnitz~\cite{roth2020counting} consider a {\em doubly restricted} version of $\Hom{G}{H}$, where both $H$ and $G$ are from {\em restricted} graph classes.
They primarily focus on the parameterized dichotomy between
poly-time solvable instances and $\#W[1]$-completeness.

We give a brief review of the graph parameters treewidth and degeneracy.
The notion of tree decomposition and treewidth were introduced in a seminal work by Robertson and Seymour~\cite{robertson1983graph,robertson1984graph,robertson1986graph}; although it has been discovered before under
different names~\cite{bertele1973non,halin1976s}.
Over the years, tree decompositions have been used extensively to 
design fast divide-and-conquer algorithms for combinatorial problems.
 Degeneracy is a nuanced measure of sparsity and has
 been known since the early work of Szekeres-Wilf~\cite{szekeres1968inequality}.
 The family of bounded degeneracy graphs is quite rich: it involves
 all minor-closed families, bounded expansion families, and preferential attachment
 graphs. Most real-world graphs tend to have small degeneracy (\cite{goel2006bounded,jain2017fast,shin2018patterns,bera2019graph,bera2020degeneracy}, also Table 2 in~\cite{bera2019graph}),
 underscoring  the practical importance of this class. The 
degeneracy has been exploited for subgraph counting problems in many
algorithmic results~\cite{chiba1985arboricity,eppstein1994arboricity,AhNe+15,jha2015path,PiSeVi17,ortmann2017efficient,jain2017fast,pashanasangi2020efficiently}.

Bressan~\cite{bressan2019faster} introduced 
the concept of DAG treewidth to design faster algorithms for 
homomorphism and subgraph counting problems in bounded degeneracy
graphs. They prove the following dichotomy for the 
subgraph counting problem. For a pattern $H$ with $|V(H)|=k$ and 
an input graph $G$ with $|E(G)|=m$ and degeneracy $\degen$,
one can count $\Hom{G}{H}$ in
$f(\degen,k)O(m^{\dtw(H)}\log m)$ time, where $\dtw(H)$ is the 
DAG treewidth of $H$ (\Cref{thm:linear}). On the other hand, assuming the exponential time hypothesis~\cite{impagliazzo1998problems}, the subgraph counting problem does not admit
any $f(\degen,k)m^{o(\dtw(H)/\ln \dtw(H))})$ algorithm, for any positive function $f:\mathbb{N}\times \mathbb{N} \rightarrow \mathbb{N}$.
Previous work of the authors shows that for every $k \geq 6$, there exists
some pattern $H$ with $k$ vertices, such that $\Hom{G}{H}$ cannot be counted
in linear time, assuming fine-grained complexity conjectures~\cite{bera2020linear}.
We note that these results do not give a complete characterization like \Thm{main}.
They define classes of $H$ that admit near-linear or specific polynomial time algorithms, and show
that \emph{some} $H$ (but not all) outside this class does not have such efficient algorithms.

We remark here that in an independent and parallel work, Gishboliner, Levanzov, and Shapira~\cite{gishboliner2020counting} effectively prove the same characterization for linear time homomorphism counting.





The problem of approximately counting homomorphism and 
subgraphs have been studied extensively in various
Big Data models such as the property testing model~\cite{eden2017approximately,eden2018approximating,assadi2018simple,eden2020faster},
the streaming model~\cite{BarYossefKS02, Manjunath2011,Kane2012, ahn2012graph,Jha2013, Pavan2013, McGregor2016, bera2017towards,bera2020degeneracy}, and the map reduce model~\cite{cohen2009graph, Suri2011, kolda2014counting}.
These works often employ clever sampling based techniques and forego exact algorithms.

Almost half a century ago, Itai and Rodeh~\cite{itai1978finding}
gave the first non-trivial algorithm for the triangle 
detection and finding problem with $O(m^{3/2})$ runtime. 
Currently, the best known algorithm for the triangle detection problem uses fast matrix multiplication and runs in time
$O(\min \{n^\omega, m^{{2\omega}/{(\omega+1)}}\})$~\cite{alon1997finding}.
Improving on the exponent is a major open problem, and it is widely believed that $m^{4/3}$ (corresponding to $\omega=2$) is a lower bound for the problem. Thus, disproving the 
\TRICONJ would require a significant breakthrough.
See~\cite{abboud2014popular} for a detailed list of other
classic graph problems whose hardness is derived using \TRICONJ.


\section{Preliminaries}
We use $G$ to denote the input graph and $H$ to denote the target pattern graph. We consider both $G = (V(G), E(G))$ and $H = (V(H), E(H))$ to be simple, undirected, and connected graphs. 
Throughout the paper, we use $m$ and $n$ to denote
$|V(G)|$ and $|E(G)|$ respectively, for the input graph $G$.
We denote $|V(H)|$ by $k$. 

We use $\HOMC_H$ to denote the problem of counting homomorphisms for a fixed pattern graph $H$. We use $\SUBC_H$ for the analogous subgraph counting problem.

If a subset of vertices $V^\prime \subseteq V(G)$ is deleted from $G$, we denote the remaining subgraph by $G-V^\prime$. We use $G[V^\prime]$ to denote the subgraph of $G$ induced by $V^\prime$.
The length of the largest induced cycle in $H$ is denoted by $\LICS(H)$.

We say that a graph $G$ is $k$\emph{-degenerate} if each non-empty subgraph of $G$ has minimum vertex degree of at most $k$. The smallest integer $k$ such that $G$ is $k$-degenerate is the \emph{degeneracy} of graph $G$. To denote the degeneracy of a graph $G$, we use $\degen(G)$ or simply $\degen$ if $G$ is clear from the context. Observe that by definition, if a graph is $k$-degenerate, all its subgraphs are also $k$-degenerate, so the degeneracy of each subgraph of $G$ is at most $\degen(G)$.
We note that the \emph{arboricity} is a closely related notion, and is related
by a constant factor to the degeneracy.

We will heavily use acyclic orientations of graphs. For the graph $H$ (say), 
we use $\Hdir$ to denote an acyclic orientation of the simple graph $H$.
Vertex orderings and DAGs are closely related to degeneracy. Given a vertex ordering $\prec$ of a graph $G$, we can obtain a DAG $G_\prec\dir$ by orienting each edge $\{u,v\} \in E(G)$ from $u$ to $v$ if 
$u\prec v$.


Now, we define a homomorphism from $H$ to $G$. We denote the number of homomorphism from $H$ to $G$ by $\Hom{G}{H}$.

\begin{definition}\label{Homomorphism}
A homomorphism from $H$ to $G$ is a mapping $\pi: V(H) \rightarrow V(G)$ such that, $\{\pi(u),\pi(v)\} \in E(G)$ for all $\{u,v\} \in E(H)$. If $H$ and $G$ are both directed, then $\pi$ should preserve the directions of the edges. If $\pi$ is injective, then it is called an embedding of $H$ in $G$.
\end{definition}
Next, we define a match (also called copy) of $H$ in $G$.
\begin{definition}\label{Ematch}
A match of $H$ in $G$ is a subgraph of $G$ that is isomorphic to $H$. If a match of $H$ is an induced subgraph of $G$, then it is an induced match of $H$ in $G$.
\end{definition}
Observe that each embedding of $H$ in $G$ corresponds to a match of $G$.

\mypar{DAG tree decompositions}
Bressan~\cite{bressan2019faster} defined the notion of {\em DAG tree decompositions}
for DAGs, analogous to the widely popular tree decompositions for undirected graphs. The crucial 
difference in this definition is that only the set of source vertices
in the DAG are considered for creating the nodes 
in the tree. Let $D$ be a DAG and $S\subseteq V$ be the set of source vertices in $D$. For a source vertex $s\in S$, let $\R_D(s)$ denote the set of vertices 
in $D$ that are reachable from $s$. For a subset of the sources $B\subseteq S$,
let $\R_D(B) = \bigcup_{s\in B} \R_D(s)$. When the underlying DAG is clear from the context, we drop the subscript $D$. 

\begin{definition}[\dtd~\cite{bressan2019faster}]
\label{def:dtd}
Let $D$ be a DAG with source vertices $S$. A \dtd of $D$ is a tree 
 $\mcT=(\mcB,\mcE)$ with the following three properties.
\begin{enumerate}
    \item Each node $B\in \mcB$ (referred to as a ``bag'' of sources) is a subset of the source vertices $S$: $B \subseteq \mcS$. \label{item:bag} 
    \item The union of the nodes in $\mcT$ is the entire set $S$: $\bigcup_{B \in \mcB} B = \mcS$. \label{item:bag_union}
    \item For all $B,~B_1,~B_2 \in \mcB$, if $B$ lies on the unique path between the nodes
    $B_1$ and $B_2$ in $\mcT$, then $\R(B_1) \cap \R(B_2) \subseteq \R(B)$. \label{item:intersection}
\end{enumerate}
\end{definition}
\noindent The {\em \dtwtext} of a DAG $D$ is then defined as the minimum over
all possible \dtd{s} of $D$,
the size of the maximum {\em bag}. For a simple undirected graph $H$,
the \dtwtext 
is the maximum \dtwtext over all possible acyclic orientations 
of $H$. We denote the \dtwtext of $D$ and $H$ by $\dtw(D)$ and 
$\dtw(H)$, respectively.

Bressan~\cite{bressan2019faster} gave an algorithm for solving
the $\HOMC_H$ problem in bounded degeneracy graphs. 
\begin{theorem}[Theorem 16 in~\cite{bressan2019faster}]
\label{thm:linear}
Given an input graph $G$ on $m$ edges with degeneracy $\degen$ and 
a pattern graph $H$ on $k$ vertices, there is an 
$O(\degen^{k}m^{\dtw(H)}\log n)$ time algorithm for solving the
$\HOMC_H$ problem.
\end{theorem}

\section{LICL and Homomorphism Counting in Linear Time} 
\label{sec:LICL-DTW-sec}
We prove that the class of graphs with $\LICS \leq 5$ is equivalent to
the class of graphs with $\dtw=1$. 

\begin{theorem}
\label{thm:LICS_DTW}
For a simple graph $H$, $\LICS(H) \leq 5$ if and only if $\dtw(H)=1$.
\end{theorem}

By \Thm{linear}, this implies that $\Hom{G}{H}$ can be determined
in near-linear time when $\LICS(H) \leq 5$ and $G$ has bounded degeneracy.

We first prove that, for a simple graph $H$, if $\LICS(H)$ is at most five, then $\dtw(H) = 1$. This is discussed in~\Cref{subsec:lics5dtw1}. Then we prove the converse: if $\LICS(H)$ is at least six, then $\dtw(H) \geq 2$. 
We take this up in~\Cref{subsec:lics6dtw2}.

\mypar{Outline of the Proof Techniques} 
Before discussing the proofs in
detail, we provide a high level description of the proof techniques.

Fix an arbitrary acyclic orientation $\Hdir$ of $H$. 
We use $\mcS$ to denote 
the set of source vertices.
We describe a recursive procedure to build a \dtd of width one,
starting from a single source in $\mcS$.

Note
that property (2) in the definition of \dtd (\Cref{def:dtd}) requires
the union of nodes in the tree to cover all the source vertices in
$\mcS$. So, we need to be careful, if we wish to use induction
to construct the final DAG tree decomposition. To this
end, we relax the property (1) and (2) of \dtd and define a notion
of {\em \pdtd}. In a \pdtd with respect to a subset $S_p\subseteq \mcS$, the nodes in the tree are subsets of $S_p$ and the union
of the nodes cover the set $S_p$. The requirement of property (3) remains
the same. The width of the tree is defined same as before. We formalize this in~\Cref{def:dtd-partial}.
Now, consider a subset $S_{r+1} \subseteq \mcS$ of size $r+1$. 
We show how to build a \pdtd of width one for $S_{r+1}$, assuming there
exists a \pdtd of width one for any subset of $\mcS$ of size $r$.

Let $x \in S_{r+1}$ and $S_{-x}$ denote the set after removing the
element $x$: $S_{-x} = S_{r+1}\setminus \{x\}$. Let $\mcT_{-x}$ be a
\pdtd of width one for the set $S_{-x}$ (such a tree exists by assumption).
We identify a ``good property'' of the tree $\mcT_{-x}$ that enables
construction of a width one \pdtd for the entire set $S_{r+1}$.
The property is the following: there exists a leaf node $\ell$ in $\mcT_{-x}$
connected to the node $d\in \mcT_{-x}$ such that $\R(x)\cap \R(\ell) \subseteq \R(d)$. We make this precise in~\Cref{def:good_pair}. 
Assume $\mcT_{-x}$ has this good property, and $\ell\in \mcT_{-x}$ be the leaf that enables $\mcT_{-x}$ to posses the good property. 
Then we first construct a width one \pdtd $\mcT_{-\ell}$ 
for the set $S_{-\ell}=S_{r+1}\setminus \{\ell\}$ and after that
add $\ell$ as a leaf node to the node $d$ in $\mcT_{-\ell}$. We prove that 
the resulting tree is indeed a valid width one 
\pdtd for $S_{r+1}$(we prove this in~\Cref{claim:pdtd}). To complete the proof, it is now sufficient to 
show the existence of an element $x\in S_{r+1}$ such that a \pdtd for 
$\mcT_{-x}$ has the good property. This is the key technical element
that distinguishes graphs with  $\LICS$ at most $5$ from those with $\LICS$ at least six. 

We make a digression and discuss this key technical element further. We consider a 
graph that captures certain reachability aspects of the source
vertices in $\Hdir$. We define this as the unique rechability
graph, $\uniqueH_{S_p}$, for a subset of the source
vertices $S_p\subseteq \mcS$. The vertex set of $\uniqueH_{S_p}$
is simply the set $S_p$. Two vertices $s_1$ and $s_2$ in $\uniqueH_{S_p}$
are joined by an edge if and only if there exists a vertex $v\in V(\Hdir)$
such that only $s_1$ and $s_2$ among the vertices in $S_p$, can reach $v$
in $\Hdir$. We prove that, if the underlying 
undirected graph $H$ has $\LICS(H) \leq 5$, then the graph $\uniqueH_{S_p}$, 
for any subset $S_{p}\subseteq \mcS$, is acyclic. 
This is given in~\Cref{lem:cycle-existence} in~\Cref{subsec:key_lemma}.

Now, coming back to the proof of~\Cref{lem:lics5dtw1}, 
we show that there must exist
an element $x\in S_{r+1}$ such that a \pdtd for 
$\mcT_{-x}$ has the ``good property'', as otherwise the unique 
reachability graph $\uniqueH_{S_{r+1}}$ over the set $S_{r+1}$
has a cycle. However, this contradicts $\LICS(H) \leq 5$ (\Cref{lem:cycle-existence}). 
This is established in~\Cref{claim:exists_good_pair}.
This completes the proof of ~\Cref{lem:lics5dtw1}.

Now consider~\Cref{lem:lics6dtw2}. Observe that, to prove this
lemma, it is sufficient to exhibit a DAG $\Hdir$ of $H$ with $\dtw(\Hdir) \geq 2$. We first prove in~\Cref{lem:triangle_to_tau}
that if the the unique reachability graph $\uniqueH_{\mcS(\Hdir)}$ 
has a triangle, then $\dtw(\Hdir) \geq 2$.
Then, for any graph $H$ with $\LICS(H) \geq 6$, we construct a
DAG $\Hdir$ such that the unique reachability graph $\uniqueH_{\mcS(\Hdir)}$ has a triangle. It follows that 
$\dtw(H)\geq 2$. 

\subsection{Main Technical Lemma}
\label{subsec:key_lemma}
In this section, we describe our main technical lemma.
We define a unique reachability
graph for a DAG $\Hdir$ over a subset of source vertices $S_{p} \subseteq \mcS(\Hdir)$. This graph captures a certain reachability aspect of the source vertices in $S_{p}$ in the graph $\Hdir$. 
More specifically, two vertices $s_1$ and $s_2$ are joined by
and edge in $\uniqueH_{S_p}$ if and only if 
there exists a vertex $v$ in $V(\Hdir)$ such that only $s_1$ and $s_2$ among the vertices in $S_p$, can reach $v$ in $\Hdir$.

\begin{definition}[Unique reachability graph]
\label{def:unique_rechability}
Let $\Hdir$ be a DAG of $H$ with source vertices $\mcS$ and
$S_p\subseteq \mcS $ be a subset of $\mcS$. 
We define a unique reachability graph 
$\uniqueH_{S_p}(S_p, E_{S_p})$ on the vertex set $S_p$, 
and the edge set $E_{S_p}$ such that there exists an edge
 $e=\{s_1,s_2\} \in E_{S_p}$, for $s_1,s_2\in S_p$,
if and only if the set $\left (\R_{\Hdir}(s_1) \cap \R_{\Hdir}(s_2) \right) \setminus \R_{\Hdir}(S_p \setminus \{s_1,s_2\})$ is non-empty.
\end{definition}

We are interested in the existence of a cycle in
$\uniqueH_{S_p}$. We show that a cycle in $\uniqueH_{S_p}$ is closely related
to an induced cycle in $H$, the underlying undirected graph of $H\dir$. 
More specifically, we prove that if $\LICS$ of $H$ is at most
five, then $\uniqueH_{S_p}$ must be acyclic for each subset
$S_p\subseteq \mcS$.

\begin{lemma} \label{lem:cycle-existence}
Let $\Hdir$ be a DAG of $H$ with source vertices $\mcS$ and
$S_p\subseteq \mcS $ be an arbitrary subset of 
$\mcS$. Let $\uniqueH_{S_p}(S_p, E_{S_p})$ be the
unique reachability graph for the subset $S_p$.
If $\LICS(H) \leq 5$, then $\uniqueH_{S_p}$ is acyclic.
\end{lemma}
\begin{proof}
We in fact prove a stronger claim. We show that if $\uniqueH_{S_p}$ has an $\ell$-cycle, then $\LICS(H) \geq 2\ell$.
Consider an edge $\{s_i,s_{j}\}$ in $E_{S_p}$. Let $\UR(s_i,s_j)$ denote
the set of vertices in $\Hdir$ reachable from $s_i$ and $s_j$ both, but 
non-reachable from any other vertices in $S_p$. Let $\dist(s,t)$ denote the length of the shortest directed path from $s$ to $t$ in $\Hdir$. We set
$\dist(s,t)=\infty$, if $t$ is not reachable from $s$. Now, let $v_{i,j}$ be the vertex in the $\UR(s_i,s_j)$ set with the smallest total
distance (directed) from $s_i$ and $s_j$ (breaking ties arbitrarily).
Formally, $v_{i,j} =  \argmin_{v} \dist(s_i,v) + \dist(s_j,v)$,
where $v \in \UR(s_i,s_j)$.

\begin{figure*}
\centering
\begin{tikzpicture}[nd/.style={scale=1,circle,draw,fill=blue,inner sep=2pt},minimum size = 6pt]    
    \matrix[column sep=3cm, row sep=1.2cm,ampersand replacement=\&]
    {
     \URCycle\\
    };
 \end{tikzpicture}
\caption{Let $S_p=\{s_1,s_2,s_3\}$. On the left, we give an example of 
a $\uniqueH_{S_p}$ graph with a triangle. On the right, we give a
possible example of the vertices $v_{1,2}$, $v_{2,3}$, and $v_{3,1}$
($v_{i,j}$ is as defined in the proof of~\Cref{lem:cycle-existence}).
$C^\prime$ forms an induced cycle of length six in $H$.}
\label{fig:URcycle}
\end{figure*}
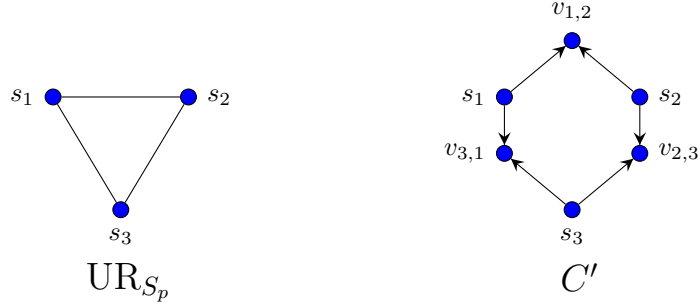

Let $C=s_1,s_2,\ldots,s_{\ell},s_1(=s_{\ell+1})$ be an $\ell$-cycle in $\uniqueH_{S_p}$.
Then using $C$ and the vertices $v_{i,i+1}$,
for $i\in [\ell]$ (abusing notation, we take $v_{\ell,\ell+1}=v_{\ell,1}$), we construct a cycle of length at least $2\ell$ in $H$. Denote by $p_{s \rightarrow v}$ the directed path from a source vertex $s \in S_p$ to a vertex  $v\in \Hdir$.
Intuitively, inserting the paths $p_{s_i \rightarrow v_{i,i+1}}$ and $p_{s_{i+1} \rightarrow v_{i,i+1}}$ between the source $s_i$ and $s_{i+1}$,
for each $i\in [\ell]$ (again, taking $s_{\ell+1}=s_1$),
induces a  cycle of length at least $2\ell$ in $H$. 
See~\Cref{fig:URcycle} for a simple demonstration of this. We make this formal below.

Let $V(p)$ denote the set of vertices of a path $p$. Any edge between vertices in $V(p_{s_i \rightarrow v_{i,i+1}})$ other than the edges of $p_{s_i \rightarrow v_{i,i+1}}$, results in either a path between $s_i$ and $v_{i,i+1}$ shorter than $\dist(s_i,v_{i,i+1})$ or a directed cycle in $H\dir$. Thus, the edges of $p_{s_i \rightarrow v_{i,i+1}}$ are the only edges between vertices in $V(p_{s_i \rightarrow v_{i,i+1}})$. 
It is easy to see that any edge between $V(p_{s_i \rightarrow v_{i,i+1}}) \setminus \{v_{i,i+1}\}$ and $V(p_{s_{i+1} \rightarrow v_{i,i+1}}) \setminus \{v_{i,i+1}\}$ result in a vertex $v^\prime_{i,i+1}$ where $\dist(s_i,v^\prime_{i,i+1})+\dist(s_{i+1},v^\prime_{i,i+1}) < \dist(s_i,v_{i,i+1})+\dist(s_{i+1},v_{i,i+1})$, therefore no such edges exist. Also, any edge from a vertex in $\R_{\Hdir}(S_p \setminus \{s_i,s_{i+1}\})$ to a vertex in $V(p_{s_i \rightarrow v_{i,i+1}})$ or $V(p_{s_{i+1} \rightarrow v_{i,{i+1}}})$ implies that $v_{i,i+1} \in \R_{\Hdir}(S_p \setminus \{s_i,s_{i+1}\})$, which is not true by definition. Hence, there are no such edges either.

We use $E(p)$ to denote the set of edges of a path $p$. 
For convenience, we assume $\ell+1$ in the index is to be treated as $1$ instead. Let
\begin{align*}
    V_{C^\prime} & = \bigcup_{p \in P} V(p), ~~~~
    E_{C^\prime}  = \bigcup_{p \in P} E(p), \\
    \text{where } P & = \bigcup_{i\in[\ell]} \{p_{s_i \rightarrow v_{i,i+1}},p_{s_{i+1} \rightarrow v_{i,i+1}}\}.
\end{align*}
Now, it is easy to see that the graph induced by the set $V_{C^\prime}$ is indeed an induced cycle $C^\prime$.
There are $2\ell$ paths in $P$, and each path $p \in P$ has at least two vertices. 
As we showed above, these paths do not share edges. Thus the length of $C^\prime$ is at least $2\ell$, and $\LICS(H) \geq 2\ell$.
Considering the contrapositive, we deduce that if $\LICS(H) \leq 5$, then $\uniqueH_{S_p}$
is acyclic.
\end{proof}

\subsection{DAG Treewidth for Graphs with LICL at most Five}
\label{subsec:lics5dtw1}

In this section, we prove the following lemma.

\begin{restatable}{lemma}{licsfivedtwone}
\label{lem:lics5dtw1}
For every simple graph $H$, if $\LICS(H) \leq 5$ then $\dtw(H) = 1$.
\end{restatable}

We introduce some notation.
We start with defining the notion of \pdtd. In this 
definition, we consider a tree decomposition with respect to a subset of the source vertices of the original DAG.
\begin{definition}[\pdtd]
\label{def:dtd-partial}
Let $\Hdir$ be a DAG with source vertices $\mcS$. For a subset $S_p \subseteq \mcS$, a \pdtd of $\Hdir$ with respect to $S_p$ is a tree $\mcT=(\mcB,\mcE)$ with the following three properties.
\begin{enumerate}
    \item Each node $B\in \mcB$ (referred to as a ``bag'') is a subset of the sources in $S_p$: $B \subseteq S_p$.
    \item The union of the nodes in $\mcT$ is the entire set $S_p$: $\bigcup_{B \in \mcB} B = S_p$.
    \item For all $B,B_1,B_2 \in \mcB$, if $B$ is on the unique path between $B_1$ and $B_2$ in $\mcT$, then we have $\R(B_1) \cap \R(B_2) \subseteq \R(B)$.
\end{enumerate}
The \pdtwtext of $\mcT$ is $\max_{B\in \mcB}|B|$. Abusing notation,
we use $\dtw(\mcT)$ to denote the \pdtwtext of $\mcT$.
\end{definition}
Observe that, when $S_p=\mcS$, we recover the original definition of 
\dtd. Our proof strategy is to show by induction on the size of the subset $S_p$ that there exists a \pdtd of width one for each $S_p \subseteq \mcS$.
In particular, when $S_p = S$, it follows that there exists a 
\dtd for $\Hdir$ of width one.

We next define \isc for a pair of vertices, based on the third 
property of the \dtd. We generalize this notion to a subset of source vertices
$S_p\subseteq \mcS$ and define a notion of \cover{$S_p$}. These notions
will be useful in identifying a suitable source vertex in an existing
\pdtd to attach a new node to it.
\begin{definition}[\isc and \cover{$S_p$}]
\label{def:cover}
Let $\Hdir$ be a DAG with sources $\mcS$. Let $s_1$ and $s_2$ be a pair of sources in $\mcS$. We call a source $s \in \mcS$ an \isc of $s_1$ and $s_2$ if $\R(s_1 )\cap \R(s_2) \subseteq \R(s)$. Furthermore, assume
$S_p \subseteq \mcS$ is a subset of the sources $\mcS$. We call a source $s \in \mcS$, a \cover{$S_p$} of $s_1\in \mcS$ if for each source $s_2\in S_p$, $s$ is an \isc for $s_1$ and $s_2$.
\end{definition}

We now introduce one final piece of notation. Assume $S_{p} \subset
\mcS$ be a subset of the source vertices in the DAG $\Hdir$.
Let $x$ be some source vertex in $\mcS$ that does not belong to
$S_p$. Let $\mcT_{S_p}$ denote a \pdtd of width one for $S_{p}$.
Now, consider a leaf node $\ell$ in $\mcT_{S_p}$. Let $d$ denote the only node 
in $\mcT_{S_p}$ that is adjacent to $\ell$. We claim that if $d$ is an \isc for 
$\ell$ and $x$, then we can construct a \pdtd for $S_{p}\cup\{x\}$ of width one (we  will make this more formal and precise in the following paragraph). We identify such source and \pdtd pair $(x,\mcT_{S_p})$ as a \good. 

\begin{definition}[\good]
\label{def:good_pair}
Let $x\in \mcS({\Hdir})$ be a source vertex and $\mcT_{S_p}$ be a \pdtd of width
one for $S_{p} \subset \mcS({\Hdir})$ where $x\notin S_p$. We call the pair $(x,\mcT_{S_p})$ a \good if there exists a leaf node $\ell \in \mcT_{S_p}$ connected to the node $d\in \mcT_{S_p}$ such that $d$ is an \isc for $x$ and $\ell$.
\end{definition}

We prove a final technical lemma that provides insight into the process 
of adding a new source vertex to an existing \pdtd of width one.
\begin{lemma} \label{lem:DTD-leaf-addition}
Let $\Hdir$ be a DAG of $H$ with sources $\mcS$ and $S_p\subset \mcS$ be a subset of $\mcS$. Assume $\mcT$ is a \pdtd for $S_p$ with $\dtw(\mcT)=1$.
Consider a source $s \in \mcS$ such that $s \notin S_p$. 
If $d\in S_p$ is a \cover{$S_p$} of $s$, then connecting $s$ to $d$ in $\mcT$ as a leaf results in a tree $\mcT^\prime$ that is a \pdtd for $S_p\cup \{s\}$. Furthermore, $\dtw(\mcT^{\prime})= 1$.
\end{lemma}
\begin{proof}
We first prove that $\mcT^{\prime}$ is a \pdtd for $S_p\cup\{s\}$.
The properties (1) and (2) of \pdtd (see~\Cref{def:dtd-partial}) trivially hold for $\mcT^\prime$. If $\mcT$ has one or two nodes, then by definition of \cover{$S_p$}, $\mcT^\prime$ satisfies property (3). So we assume $\mcT$ has at least $3$ nodes.

Note that $\mcT$ and $\mcT^{\prime}$ are identical barring the leaf node $s$.
Hence, for any three nodes $s_1$, $s_2$, and $s_3$ in $T^{\prime}$ 
with $s \notin \{s_1,s_2,s_3\}$, property (3) of \pdtd (\Cref{def:dtd-partial})
holds. Now, consider $s$ with two other nodes $s_1$ and $s_2$ in $\mcT^\prime$ where
$s_1$ is on the unique path between $s$ and $s_2$.
If $s_1=d$, then property (3) holds as $d$ is a \cover{$S_p$} of $s$.
So assume $s_1\neq d$. But then, $s_1$ is on the unique path between $d$
and $s_2$ (by construction of $\mcT^{\prime}$). Since
property (3) holds for $d$, $s_1$, and $s_2$ in $\mcT$, we have
$\R(d) \cap \R(s_2) \subseteq \R(s_1)$. We also have 
$\R(s) \cap \R(s_2) \subseteq \R(d)$ as $d$ is a \cover{$S_p$} of $s$.
Hence, $\R(s) \cap \R(s_2) \subseteq \R(s_1)$. Therefore, property (3) holds. Thus, $\mcT^\prime$ is a \pdtd of $S_p\cup \{s\}$. As $\dtw(\mcT)=1$, it follows immediately from the construction that $\dtw(\mcT^{\prime})=1$.
\end{proof}

We now have all the ingredients to prove~\Cref{lem:lics5dtw1}.
For the sake of completeness, we restate the lemma.
\licsfivedtwone*

\begin{proof}
The \dtwtext of a simple graph $H$ is defined as
the maximum \dtwtext of any DAG $\Hdir$ obtained from $H$. 
So, we prove that $\dtw(\Hdir)=1$ for each DAG $\Hdir$ of $H$.
In the rest of the proof, we fix a DAG of $H$, and call it $\Hdir$.
Let $\mcS({\Hdir})$ denote the set of all source vertices in $H\dir$. When $\Hdir$ is clear from the context, we simply use $\mcS$.

Let $S_p \subseteq \mcS$ denote a subset of $\mcS$. 
We prove by induction on the size of the subset $S_p$ that
there exists a \pdtd (\Def{dtd-partial}) of width one for each $S_p 
\subseteq \mcS$. In particular, when $S_p = S$, it follows that there exists
a \dtd for $\Hdir$ of width one.

The base case of $|S_p|=1$ is trivial: put the only source in $S_p$
in a bag $B$ as the only node in the \pdtd for $\Hdir$. 
Similarly, for $|S_p|=2$, put the two sources in two separate bags and
connect them by an edge in the \pdtd for $\Hdir$. The resulting tree
is a \pdtd of width one.
Now assume that, it is possible to build a \pdtd of width one 
for any subset $S_p\subset \mcS$ where $|S_p| \leq r$, and $ 1 \leq r < |\mcS|$. We show how to construct a \pdtd of width one for any subset $S_p\subseteq \mcS$ where $|S_p|=r+1$ for $r\geq 2$. 

Fix a subset $S_{r+1}\subseteq \mcS$ of size $r+1$. 
Consider an arbitrary source $x \in S_{r+1}$. By induction hypothesis, we can construct a \pdtd of width one for the set $S_{-x} = S_{r+1} \setminus \{x\}$. We call the tree $\mcT_{-x}$. 
Now recall that, we call the pair $(x,\mcT_{-x})$ a \good if there exists
a leaf node $\ell \in \mcT_{-x}$ connected to the node $d\in \mcT_{-x}$ such that
$d$ is an \isc for $x$ and $\ell$ (see~\Def{good_pair}).
We argue the existence of a \good $(x,\mcT_{-x})$ and give a
constructive process to find a width one \pdtd of $S_{r+1}$
from such a \good $(x,\mcT_{-x})$.

\mypar{A \good leads to a \pdtd of width one}
We first show that if there exists a source $x\in S_{r+1}$ and a
width one \pdtd $\mcT_{-x}$ for $S_{-x} = S_{r+1}\setminus \{x\}$ such that $(x,\mcT_{-x})$
is a \good, then there exists
a width one \pdtd for $S_{r+1}$. In fact, we give a simple
constructive process to find such a \pdtd: construct a width one \pdtd $\mcT_{-\ell}$ for $S_{-\ell}=S_{r+1} \setminus \{\ell\}$, and then connect $\ell$ as a leaf to $d$ in $\mcT_{-\ell}$. 

\begin{claim}
\label{claim:pdtd}
Let $x\in S_{r+1}$ be a source vertex and $\mcT_{-x}$ be a 
width one \pdtd for $S_{-x} = S_{r+1}\setminus \{x\}$ such that $(x,\mcT_{-x})$
is a \good. Then, there exists a \pdtd $\mcT$ for $S_{r+1}$ with $\dtw(\mcT)=1$.
\end{claim}
\begin{proof}
Since $(x,\mcT_{-x})$ is a \good, there exists a
leaf node $\ell \in \mcT_{-x}$ connected to the node $d \in \mcT_{-x}$
such that $d$ is an \isc for $x$ and $\ell$.
We build a \pdtd of width one for $S_{-\ell} = S_{r+1} \setminus \{\ell\}$ 
(such a tree exists by induction hypothesis), 
and then add $\ell$ as a leaf node to the node $d$. 
We prove that the resulting tree, denoted as $\mcT$, is \pdtd for $S_{r+1}$ with
$\dtw(\mcT)=1$.

Since $\ell$ is only connected to $d$ in $\mcT_{-x}$,
$d$ is a \cover{$S_{-x}$} of $\ell$.
Also, $d$ is an \isc of $\ell$ and $x$, so $d$ is a \cover{$S_{r+1}$} of $\ell$. Therefore, by applying~\Cref{lem:DTD-leaf-addition}, it follows that $\mcT$ is a \pdtd 
of $S_{r+1}$ with $\dtw(\mcT)=1$.
\end{proof}

\mypar{Existence of a \good}
We have shown how to construct a \pdtd for the set $S_{r+1}$
if there exists a \good $(x,\mcT_{-x})$ where $x$ is a source in $S_{r+1}$ and
$\mcT_{-x}$ is a width one \pdtd for $S_{-x}$.
We now show that for any set $S_{r+1}$, there always exists
a \good $(x,\mcT_{-x})$.

\begin{claim}
\label{claim:exists_good_pair}
There exists a vertex $x\in S_{r+1}$ and a width one \pdtd $\mcT_{-x}$ for
$S_{-x} = S_{r+1} \setminus \{x\}$, such that $(x,\mcT_{-x})$ is a \good.
\end{claim}
\begin{proof}
Assume for contradiction, the claim is false. Consider the  unique reachability graph on the vertex set $S_{r+1}$, denoted by
$\uniqueH_{S_{r+1}}(S_{r+1},E_{S_{r+1}})$ (see~\Def{unique_rechability}). Let $x \in S_{r+1}$ be an arbitrary source vertex. 
By assumption, $(x,\mcT_{-x})$ is not a \good. 
So, for each leaf node $\ell\in \mcT_{-x}$ connected to the node $d\in \mcT_{-x}$, 
$d$ is not an \isc for $x$ and $\ell$. 
Then, there exists a vertex $v$ in $\Hdir$,
such that $v\in \R(x)\cap \R(\ell)$, but $v\notin \R(d)$. On the other hand,
by construction, $d$ is a \cover{$S_{-x}$} for $\ell$ ($d$ is the only node connected
to $\ell$ in $\mcT_{-x}$).
Hence, $v$ is reachable from none of the source vertices in $S_{-x}$, other than $\ell$. 
Therefore, the edge $\{x,\ell\} \in E_{S_{r+1}}$. 
Now, $\mcT_{-x}$ has at least two leaves, so the degree of the source vertex $x$ in $\uniqueH_{S_{r+1}}$ is at least 2. The same argument holds for each $x\in S_{r+1}$. 
Hence, the degree of each vertex in $\uniqueH_{S_{r+1}}$ is at least two. Now $|S_{r+1}| \geq 3$ (recall $r\geq 2$), thus there exists a cycle $\mcC$ in $\uniqueH_{S_{r+1}}$ of length at least $3$. 
By applying~\Cref{lem:cycle-existence}, we have $\LICS(H) \geq 6$.
But $\LICS(H) \leq 5$, so this leads to a contradiction. Hence, the claim is true.
\end{proof}
We proved by induction that for any non-empty subset $S_p\subseteq \mcS$, there exist a \pdtd for $S_p$ with width one. In the case when $S_p = \mcS$, the \pdtd is a \dtd for $H\dir$. This completes the proof of~\Cref{lem:lics5dtw1}.
\end{proof}

\subsection{DAG Treewidth for Graphs with LICL at least Six}
\label{subsec:lics6dtw2}
In this section, we prove the following lemma.

\begin{restatable}{lemma}{licssixdtwtwo}
\label{lem:lics6dtw2}
For every simple graph $H$, if $\LICS(H) \geq 6$ then $\tau(H) \geq 2$.
\end{restatable}

We first discuss 
the simple case of the $6$-cycle. Note that, to prove that $\dtw(H) \geq 2$, it suffices to show that there exists a DAG $\Hdir$ of $H$ such that $\dtw(\Hdir) \geq 2$. Let $\Hdir$ be a DAG of $H$ as shown in the middle figure in \Fig{6-cycle}. Let $\mcS=\{s_1,s_2,s_3\}$ be the set of sources in $\Hdir$. Consider the unique reachability graph $\uniqueH_{\mcS}(\mcS,E_{\mcS})$, shown on the right in~\Fig{6-cycle}. 
The graph $\uniqueH_{\mcS}$ is a triangle: $t_1$ is not reachable from $s_1$, but reachable from $s_2$ and $s_3$, and so on. 
In any \dtd $\mcT$ of $\Hdir$ with width one, all source vertices are a vertex of $\mcT$ by themselves. So at least one of $s_1$, $s_2$, or $s_3$  (say $s_1$) would be on the unique path between the other two. 
But this would violate property (3) of \dtd.
It follows that $\dtw(H) \geq 2$.
In this case, it is not difficult to argue that $\dtw(H)=2$.

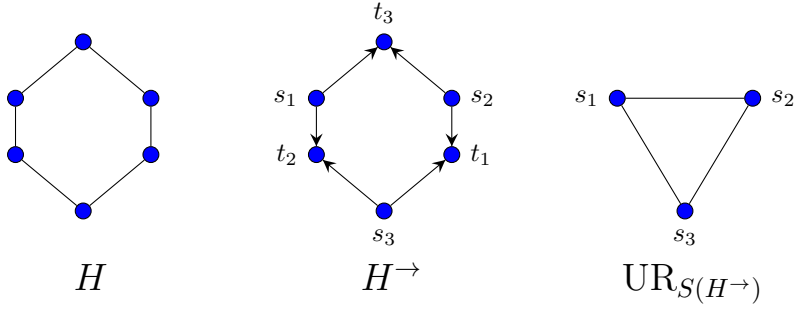
\begin{figure*}
\centering
\begin{tikzpicture}[nd/.style={scale=1,circle,draw,fill=blue,inner sep=2pt},minimum size = 6pt]    
    \matrix[column sep=1cm, row sep=0.2cm,ampersand replacement=\&]
    {
     \DTDSixCycle\\
    };
 \end{tikzpicture}
\caption{Let $H$ be a six cycle. In the middle figure, we show the DAG $\Hdir$ of $H$ for which $\dtw(\Hdir) \geq 2$. On the right, we show the $\uniqueH$ graph corresponding to $\Hdir$. It is a triangle as $t_1$ is only reachable from $\{s_2,s_3\}$, 
$t_2$ is only reachable from $\{s_1,s_3\}$, and
$t_3$ is only reachable from $\{s_1,s_2\}$.}
\label{fig:6-cycle}
\end{figure*}

We formalize this intuition to prove in the following lemma:
if the $\uniqueH_{\mcS}$ graph of a DAG $\Hdir$ with source vertices $\mcS$ has a triangle in it, then it must the case that $\dtw(\Hdir) \geq 2$. Then, to prove $\dtw(H)\geq 2$ for a graph $H$, it is 
sufficient to show the existence of a DAG $\Hdir$ such that
the corresponding $\uniqueH_{\mcS(\Hdir)}$ has a triangle.

\begin{lemma}
\label{lem:triangle_to_tau}
Let $\Hdir$ be a DAG of $H$ with source vertices $\mcS$ and
$\uniqueH_{\mcS}(\mcS, E_{\mcS})$ be the
{\em unique reachability} graph for the set $\mcS$.
If $\uniqueH_{\mcS}$ has a triangle, then $\dtw(H) \geq 2$.
\end{lemma}
\begin{proof}
Assume for contradiction, $\dtw(\Hdir) =1$ and
$\mcT$ be a \dtd of width one. Let $\{s_1,s_2,s_3\}$ be a triangle 
in the graph $\uniqueH_{\mcS}$. As $\mcT$ is a \dtd with width one, 
all source vertices in $\mcS$ must be a node by themselves  in $\mcT$.
Observe that, there must exists a node $s\in \mcT$ that is
between the unique path for a pair of nodes in $\{s_1,s_2,s_3\}$. As otherwise, these three nodes are all pairwise connected by an 
edge forming a triangle in $\mcT$.
Wlog, assume $s$ is on the unique path between $s_1$ and $s_2$ in $\mcT$. Since, $\{s_1,s_2\}$ is an edge in $\uniqueH_{S}$, by definition, there exists a vertex $t_{s_1,s_2}\in V(\Hdir)$, such that $t\in \R_{\Hdir}(s_1) \cap \R_{\Hdir}(s_2)$, but $t\notin \R_{\Hdir}(s)$. But, this violates the property (3) of \dtd in~\Cref{def:dtd}.
So such a tree $\mcT$ cannot exists and hence, $\dtw(\Hdir) \geq 2$. Therefore, $\dtw(H) \geq 2$.
\end{proof}

We are now ready to prove the main lemma. We restate the lemma for completeness.
\licssixdtwtwo*
\begin{proof}
Let $\LICS(H)=r$, where $r \geq 6$ and
$r=3\ell+q$, for some $\ell \geq 2$ where $q\in\{0,1,2\}$.
Assume ${C}=v_1,v_2,\ldots,v_r,v_1$ 
is an induced cycle of length $r$ in $H$.
We construct a DAG $\Hdir$ as follows.

Consider an edge $e=(u,v)$ in $H$.
Assume only one of the end point  does not belong to $V(C)$ --- say $u\notin V(C)$. Then we orient the edge from $v$ to $u$. 
Now consider the case when both $u,v\notin V(C)$. We orient such 
edges in an arbitrary manner ensuring the resulting orientation is acyclic. 
We now describe the orientation of the edges on the $r$-cycle $C$.
We mark three vertices $s_1$, $s_2$ and $s_3$ in ${C}$ as sources that are at least distance two apart from each other. Wlog, assume $s_1=v_1$, $s_2=v_{\ell+1}$, and $s_3=v_{2\ell+1}$.
Now we mark three vertices $t_1$, $t_2$, and $t_3$ as sinks such that 
$t_1$ is between $s_1$ and $s_2$, $t_2$ is between $s_2$ and $s_3$, and $t_3$ is between $s_3$ and $s_1$ in the cycle $C$. 
Again, wlong assume $t_1 = v_2$, $t_2=v_{\ell+2}$, and $t_3=v_{2\ell+2}$.
Finally, orient the edges in ${C}$ towards the sink vertices and away from the sources. This completes the description of $\Hdir$.

Now let $\mcS$ denote the set of source vertices in $\Hdir$.
Consider the unique reachability graph $\uniqueH_{\mcS}(S,E_{\mcS})$.
We claim that $\uniqueH_{S}$
includes a triangle. Indeed, we show that 
$\{s_1,s_2,s_3\}$ forms a triangle in $\uniqueH_{\mcS}$.
We first argue the existence of the edge $\{s_1,s_2\}\in E_{S}$. The vertex $t_1$ is reachable from $s_1$ and $s_2$, but not from $s_3$. Since all the edges that are not
part of the cycle ${C}$, are oriented outwards from the
vertices in ${C}$, no other source vertices in $\mcS$
can reach $t_1$ in $\Hdir$. Hence, $\{s_1,s_2\}\in E_{S}$.
Similarly, we can argue the existence of the edges 
$\{s_2,s_3\}$ and $\{s_1,s_3\}$ in $E_{S_p}$. Applying~\Cref{lem:triangle_to_tau}, it follows that $\dtw(H)\geq 2$.
\end{proof}

\section{LICL and the Homomorphism Counting Lower Bound} 
\label{sec:lower-bound_sec}
In this section, we prove our main lower bound result. We show that 
for a pattern graph $H$ with $\LICS(H) \geq 6$, the $\HOMC_H$ problem does not 
admit a linear time algorithm in bounded degeneracy graphs, assuming the \TRICONJ (\Cref{conj:triangle}). We state our main theorem below.

\begin{restatable}{theorem}{LB}
\label{thm:lower_bound}
Let $H$ be a pattern graph on $k$ vertices with $\LICS(H) \geq 6$. Assuming  the \TRICONJ, there exists an absolute constant $\gamma > 0$ such that for any function $f:\mathbb{N \times N} \to \mathbb{N}$, there is no (expected) $f(\degen, k)o(m^{1+\gamma})$ algorithm for the $\HOMC_H$ problem, where $m$ and $\degen$ are the number of edges and the degeneracy of the input graph, respectively.
\end{restatable}

\mypar{Outline of the Proof}
We first present an outline of our proof; the complete proof is discussed in~\Cref{subsec:lb_proof}.
Let $\TRIC$ denote the problem of counting the number of triangles in a graph. We prove the theorem by a linear time Turing reduction from the $\TRIC$ problem to the $\HOMC_H$ problem.
Assuming the \TRICONJ, any algorithm (possibly randomized) for the $\TRIC$ problem requires $\omega(m)$ time, where $m$ is the number of edges in the input graph. Given an input instance $G$ of the \TRIC problem, we construct a graph $G_H$ with bounded degeneracy and $O(|E(G)|)$ edges.
We show how the number of triangles in $G$ can be obtained by counting specific homomorphisms of $H$ in $G_H$.

Let $\LICS(H)=r$ and $C$ be one of the largest induced cycles in $H$; let $V(C)$ denotes
its vertices.
We now describe the construction of the graph $G_H$.
The graph $G_H$ has two main parts: (1) the fixed component, denoted as $G_{\fixedGraph}$ (this part is independent of the input graph $G$ and only depends on the pattern graph $H$) and (2) the core component, denoted as 
$G_{\coreGraph}$. Additionally, there are edges that connect these two components, denoted by $E_{\connect}$.
Let $H_{\excludeC}$ denote the graph after we remove $V(C)$ from $H$. More formally, $H_{\excludeC} = H-V(C)$. The fixed component $G_{\fixedGraph}$ is a copy of $H_{\excludeC}$.

Next, we give an intuitive account of our construction. We discuss the role of $G_{\coreGraph}$ and how its connection to $G_{\fixedGraph}$ through $E_{\connect}$
ensures that the number of triangles in $G$ can be obtained by counting homomorphisms of $H$. Then we give an overview of the construction.

\mypar{Intuition behind the Construction}
The main idea is to construct $G_{\coreGraph}$ and $E_{\connect}$ in such a way that each triangle in $G$ transforms to an $r$-cycle in $G_{\coreGraph}$, that then composes a match of $H$ together with $G_{\fixedGraph}$ (recall that $G_{\fixedGraph}$ is a copy of $H_{\excludeC}$). To this end, we design $G_{\coreGraph}$ in $r$ parts, ensuring the following properties hold for each $r$-cycle in $G_{\coreGraph}$ that contains exactly one vertex in each of these $r$ parts: (1) It composes a match of $H$ together with $G_{\fixedGraph}$ and (2) It corresponds to a triangle in $G$. Let $\mcP^\prime$ denote the partition of $V(G_{\coreGraph})$ into these $r$ parts. Further, assume we construct $G_H$ in a way that, each match of $H$ that contains the vertices of $G_{\fixedGraph}$ and exactly one vertex in each set $V \in \mcP^\prime$, corresponds to one of the $r$-cycles described above. It is now easy to see that if we can count these matches of $H$, we can then obtain the number of the described $r$-cycles in $G_{\coreGraph}$ and hence the number of triangles in $G$.

Consider a partition $\mcP$ of $V(G_H)$ where $|\mcP| = k$. Assume there is a linear time algorithm $\ALG$ for $\HOMC_H$ in bounded degeneracy graphs. Then,~\Lem{hom-inc-exc} proves that there exists a linear time algorithm that, using $\ALG$, counts the matches of $H$ in $G$ that include exactly one vertex in each set $V \in \mcP$. These matches are called $\pmatch$es of $H$, as we define formally later in~\Def{p-match}. Also, each $r$-cycles in $G_{\coreGraph}$ that contain exactly one vertex in each set $V\in \mcP^\prime$ is a $\ppmatch$ of $C$. Now, we define the partition $\cP$ of $V(G_H)$ as follows; $\cP$ includes each set in $\mcP^\prime$ and each of the $k-r$ vertices in $G_{\fixedGraph}$ as a set by itself. 

Overall, by construction of $G_H$, we can get the number of triangles in $G$ by the number of $\ppmatch$es of $C$ in $G_{\coreGraph}$. Further, we can obtain the number of $\ppmatch$es of $C$ in $G_{\coreGraph}$ by the number of $\pmatch$es of $H$ in $G_H$ that we count using $\ALG$. The following restates the desired properties of $G_H$ we discussed, more formally.

\begin{enumerate}[label=(\Roman*)]
    \item There is a bijection between the set of $\pmatch$es of $H$ in $G_H$ and the set of $\ppmatch$es of $C$ in $G_{\coreGraph}$.
    \item The number of triangles in $G$ is a simple linear function of the number of $\ppmatch$es of $C$ in $G_{\coreGraph}$.
\end{enumerate}

Next, we give an overview of the construction of $G_H$ for $r=6$. We prove that properties (I) and (II) hold for our construction in the general case in~\Lem{copy-cycle-bijection} and~\Lem{triangle-linear-func}, respectively. 

\mypar{Overview of the Construction}
In what follows, we give an overview of $G_{\fixedGraph}$, $G_{\coreGraph}$, and $E_{\connect}$. For the ease of presentation, we assume $r=6$.
\begin{enumerate}[label=(\arabic*)]
    \item $G_{\fixedGraph}$ is a copy of $H_{\excludeC}$.
    We denote the vertex set in $G_{\fixedGraph}$
    as $V_{\fixed}$.
    Observe that, $G_{\fixedGraph}$ does not
    depend on the input graph $G$.
    
    \item The core component $G_{\coreGraph}$ consists of two set
    of vertices: $V_{\core}$ and $V_{\aux}$. We first discuss
    the sets $V_{\core}$ and $V_{\aux}$, and then introduce the edge set $E(G_{\coreGraph})$.
    \begin{enumerate}
        \item $V_{\core}$ consists of three set of vertices, 
        $V_1 = \{w_1,\ldots,w_n\}$, 
        $V_2 = \{x_1,\ldots,x_n\}$, and 
        $V_3 = \{y_1,\ldots,y_n\}$ --- each of size $n$ (recall $|V(G)|=n$). The vertices in each of these sets correspond to the vertices in $V(G)=\{u_1,\ldots, u_n\}$.
    
        \item The construction of $V_{\aux}$ depends on $r$. For $r=6$,
            it consists of three sets, denoted as $V_{1,2}$, $V_{2,3}$, and $V_{1,3}$ --- each of size $2m$ (recall $|E(G)|=m$). 
            The vertices in each of these sets corresponds to the edges in $E(G)$. We index them using $e$, for each $e\in E(G)$:
            $V_{1,2}=\{v_e^{1,2},v_e^{2,1}\}_{e\in E(G)}$, and so on.
            The role of these sets will become clear as we describe the edges of $G_{\coreGraph}$.
        \item Consider an edge $e=\{u_i,u_j\}\in E(G)$ 
        and the pair $V_1$ and $V_2$. We connect the 
        vertex $w_i \in V_1$ to the vertex $x_j \in V_2$ by a $2$-path via the vertex $v^{1,2}_e \in V_{1,2}$. Similarly, we connect the vertex $w_j$ to the vertex $x_i$ by a $2$-path via the vertex $v^{2,1}_e$. In particular, we add the
        edges $\{w_i,v^{1,2}_e\}$ and $\{v^{1,2}_e,x_j\}$, and the edges 
        $\{w_j,v^{2,1}_e\}$ and $\{v^{2,1}_e,x_i\}$ to the set $E(G_{\coreGraph})$.
        We repeat the process for the pairs $(V_2,V_3)$ and $(V_1,V_3)$
        for each edge $e\in E(G)$.
    \end{enumerate}
    
    \item We now describe the edge set $E_{\connect}$ that serves as 
    connections between $G_{\fixedGraph}$ and $G_{\coreGraph}$.
    Let $\map_{\connect}$ be a bijective mapping between the two sets $V(C)$ and
    $\{V_1,V_2,V_3,V_{1,2},V_{2,3},V_{1,3}\}$; 
    $\map_{\connect}: V(C) \rightarrow \{V_1,V_{1,2},V_2,V_{2,3},V_3,V_{1,3}\}$.
    For each edge $e=\{u,v\} \in E(H)$ such that $u\in V(C)$
    and $v\notin V(C)$, we do the following.
    Let $z_v\in V_{\fixed}$ denote the vertex 
    corresponding to the vertex $v$ (recall $G_\fixed$ is a copy of $H_{\excludeC}$). We connect $z_v$
    to all the vertices in the set $\map_{\connect}(u)$ and add these
    edges to $E_{\connect}$.
\end{enumerate}

Note that, here $\mcP^\prime = \{V_1,V_{1,2},V_2,V_{2,3},V_3,V_{1,3}\}$. Before diving into the details of deriving the triangle counts in $G$, we first take an example pattern graph $H$ to visually depict the constructed graph $G_H$ (see~\Cref{fig:H_and_G_H-example}) and discuss why properties (I) and (II) hold in our construction.

\mypar{An Illustrative Example}
Let $H$ be the graph as shown in~\Cref{fig:H-example}. In this example, $\LICS(H)=6$. Let $C=a_3,a_4,a_5,a_6,a_7,a_8,a_3$ be the 
induced $6$-cycle in $H$. We demonstrate the constructed graph $G_H$ in ~\Cref{fig:G_H-example}. We now discuss the various components of $G_H$.
\begin{enumerate}[label=(\arabic*)]
    \item The graph $G_{\fixedGraph}$ is shown by the red oval.
The vertices $z_1$ and $z_2$ compose $V_{\fixed}$, where $z_1$ corresponds to $a_1$ and $z_2$ corresponds to $a_2$. 

    \item The graph $G_{\coreGraph}$ is shown by the blue oval.
For each edge $e=\{u_i,u_j\}\in E(G)$, (for some input graph $G$, which
is not shown in the figure), we add a total of six 2-paths: two between each pair of sets from $\{V_1,V_2,V_3\}$. For instance, between the set
$V_1$ and $V_2$ these 2-paths are as follows: $\{w_i,v_{e}^{1,2},x_j\}$
and $\{w_j,v_{e}^{2,1},x_i\}$. The vertices $w_i,w_j$ belong to $V_1$;
$x_i,x_j$ belong to $V_2$; and $v_{e}^{1,2},v_{e}^{2,1}$ belong to $V_{1,2}$.

    \item Finally we describe the edge set $E_{\connect}$ (the edges in violet).
We consider the following bijective mapping $\map_{\connect}$:
$\map_{\connect}(a_3)=V_1$, $\map_{\connect}(a_4)=V_{1,2}$,
$\map_{\connect}(a_5)=V_2$, $\map_{\connect}(a_6)=V_{2,3}$,
$\map_{\connect}(a_7)=V_3$, $\map_{\connect}(a_8)=V_{1,3}$.
Now consider the edge $\{a_3,a_1\} \in E(G)$;
$a_3 \in V(C)$ and $a_1\notin V(C)$. So we connect $z_1$
(the vertex corresponding to $a_1$) to each
vertex in the set $\map_{\connect}(a_3)=V_1$. We repeat the
same process for each edge $\{u,v\}$ in $E(H)$ where $u\in V(C)$
and $v\notin V(C)$.
\end{enumerate}

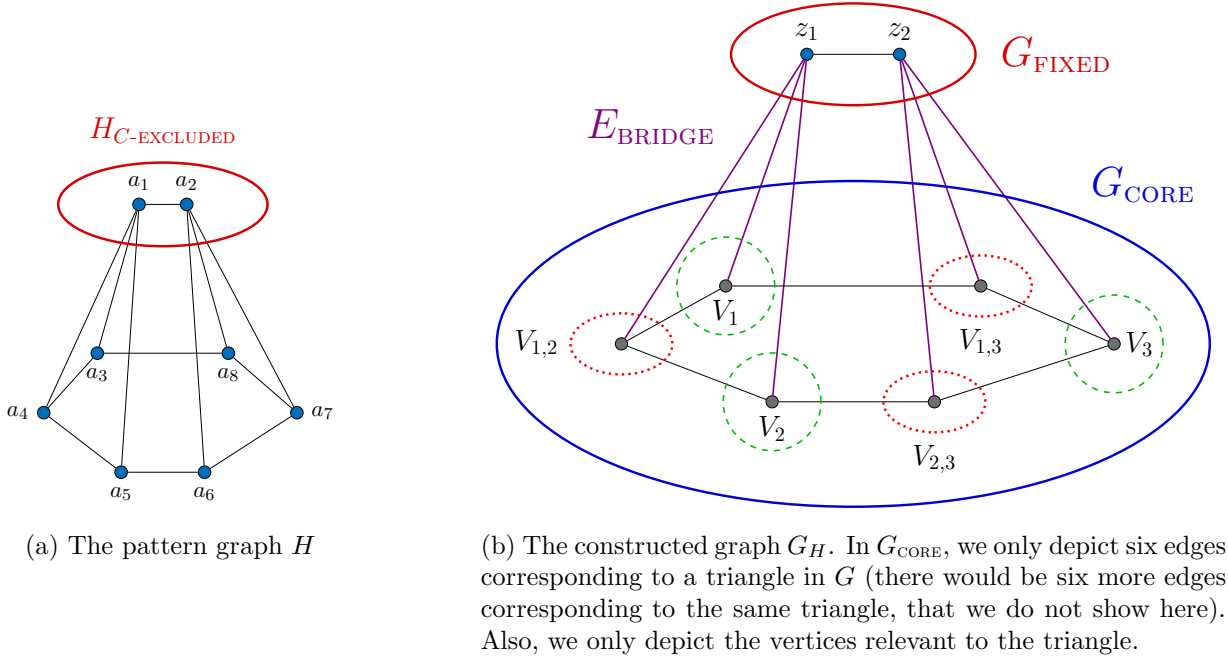
\begin{figure}[t]
\centering
\begin{subfigure}[t]{0.3\textwidth}
\centering
\resizebox{\textwidth}{!}{
  \begin{tikzpicture}[nd/.style={scale=1,circle,draw,fill=NavyBlue,inner sep=2pt},minimum size = 6pt]
    \matrix[column sep=0.8cm, row sep=0.5cm,ampersand replacement=\&]
    {
        \exampleH \\
    };
  \end{tikzpicture}
  }
\caption{The pattern graph $H$}
\label{fig:H-example}
\end{subfigure}\hfill
\begin{subfigure}[t]{0.6\textwidth}
\centering
\resizebox{\textwidth}{!}{
  \begin{tikzpicture}[nd/.style={scale=1,circle,draw,fill=Black!70,inner sep=2pt},minimum size = 6pt]
    \matrix[column sep=0.8cm, row sep=0.5cm,ampersand replacement=\&]
    {
        \GHexampleH \\
    };
  \end{tikzpicture}
  }
\caption{The constructed graph $G_H$. In $G_{\coreGraph}$, we only depict six edges corresponding to a triangle in $G$ (there would be six more edges corresponding to the same triangle, that we do not show here). Also, we only depict the vertices relevant to the triangle.}
\label{fig:G_H-example}
\end{subfigure}
\caption{$G_H$ constructed for an example pattern graph $H$}
\label{fig:H_and_G_H-example}
\end{figure}

Observe that in this example, $\mcP = \{\{z_1\},\{z_2\},V_1,V_{1,2},V_2,V_{2,3},V_3,V_{1,3}\}$. It is easy to see that $E_{\connect}$ connects $G_{\coreGraph}$ to $G_{\fixedGraph}$, such that each $6$-cycle in $G_{\coreGraph}$ compose a match of $H$ together with $G_{\fixedGraph}$. Each $\pmatch$ of $H$ is actually an induced match as the only edges between its vertices in $G_H$ are the edges that correspond to the match. Therefore, in this example, each $\pmatch$ of $H$ in $G_H$ include a $6$-cycle in $G_{\coreGraph}$ that is actually a $\ppmatch$ of $C$. Thus, property (I) holds.

It is not difficult to see that a triangle in $G$ introduces a total of 
six many $6$-cycle in $G_{\coreGraph}$ that are $\ppmatch$es of $C$ in $G_{\coreGraph}$. The converse follows as each $\ppmatch$ of $C$, which is a $6$-cycle in $G_{\coreGraph}$, must contain exactly one vertex from each of the three sets in each of $V_{\core}$ and $V_{\aux}$. So, we could obtain the number of triangles in $G$ by dividing the number of $\ppmatch$es of $C$ in $G_{\coreGraph}$ by six. Thus, property (II) holds.

\mypar{Deriving The Triangle Counts in \bm{$G$}}
So far, we have shown that properties (I) and (II) hold in $G_H$ for our construction. Therefore, the number of $\pmatch$es of $H$ in $G_H$ reveals the number of triangles in $G$. However, we are interested
in utilizing the homomorphism count of $H$ to derive the triangle count in $G$. Indeed, we obtain the number of $\pmatch$es of $H$ in $G_H$ by carefully looking at ``restricted'' homomorphisms from $H$ to $G$. One crucial property of the graph $G_H$ that we will require is bounded degeneracy. In fact, our construction of the graph $G_H$ ensures that it has constant degeneracy irrespective of the degeneracy of $G$ (we will formally prove this later in~\Cref{lem:bounded-degeneracy}). 

Let $\ALG$ be an algorithm for the $\HOMC_H$ problem, that runs in $f(\degen,k)\cdot O(m)$ time for some explicit function $f$, where $m$ and $\degen$ are the number of edges and degeneracy of the input graph, respectively. Then, we can use $\ALG$ to count the homomorphisms from $H$ to any subgraph of $G_H$ in time $f(\degen(G_H),k)\cdot O(m)$. Note that,
here we use the fact that for any subgraph $G_{H}^{\prime}$ of $G_H$,
$\degen(G_{H}^{\prime}) \leq \degen (G_{H})$.

We now solve the final missing piece of the puzzle:
how to count the number of $\pmatch$es of $H$ in $G_H$ using $\ALG$?
We present a two step solution to this question.
First, we count the number of ``${\cP}$ restricted" homomorphisms, denoted
by $\phom$ and defined in ~\Cref{def:p-match},
from $H$ to $G_H$ by running $\ALG$ on carefully chosen subgraphs of
$G_H$. Intuitively, a ``${\cP}$ restricted" homomorphism
is a homomorphism from $H$ to $G_H$ that involves at least one vertex in each part of ${\cP}$. Second, we use the count from the first step to 
derive the number of $\pmatch$es of $H$ in $G_H$.
We present this in~\Cref{lem:hom-inc-exc}.

We now formally define $\pmatch$ and $\phom$.
\begin{definition}[$\pmatch$ and $\phom$]
\label{def:p-match}
Let ${\cP}=\{V_1,\ldots,V_k\}$ be a partition of the vertex set $V(G)$
of the input graph $G$ where $|V(H)|=k$ for the pattern graph $H$. Further assume $|V_i| \geq 1$ for each $i\in [k]$.
Let $G_{\Hmatch}$ be a subgraph of $G$ such that $G_{\Hmatch}$
is a match of $H$. We call $G_{\Hmatch}$ a $\pmatch$, if it
includes exactly one vertex from each set $V_i$ in ${\cP}$:
$|V(G_{\Hmatch})\cap V_i| = 1$ for each $i\in [k]$.
Let $\pi:V(H) \rightarrow V(G)$ be a homomorphism from $H$ to $G$.
We call $\pi$ a $\phom$, if the image of $\pi$ is non-empty in each
set $V_i$: $|\{v:\pi(u)=v~\text{ for } u\in V(H)\} \cap V_i| \geq 1$
for each $i\in [k]$.
\end{definition}
In the following lemma, we prove that it is possible to count the number of 
$\pmatch$es of $H$ in $G_H$ by running $\ALG$ on suitably chosen $2^k$ many
subgraphs of $G_H$.

\begin{lemma} \label{lem:hom-inc-exc}
Assume that $\ALG$ is an algorithm for the $\HOMC_H$ problem
that runs in time $O(mf(\degen,k))$ for some function $f$, where $m=E(G)$ and $\degen=\degen(G)$ for the input graph $G$, and $k=V(H)$. 
Let ${\cP}=\{V_1,\ldots,V_k\}$ be a partition of $V(G)$ with $|V_i|\geq 1$ for each $i\in [k]$. 
Then, there exists an algorithm that counts the
number of $\pmatch$ of $H$ in $G$ with running time $O(2^k \cdot mf(\degen,k))$.
\end{lemma}

\begin{proof}
Let ${\cF}_1,{\cF}_2,\ldots,{\cF}_{2^k-1}$ be the non-empty subfamilies of ${\cP}$. Let $G_1, G_2, \ldots, G_{2^k-1}$ be the subgraphs of $G$ where $G_i$ is
induced on the vertex set $V(G)\setminus (\bigcup_{S \in {\cF}_i}S)$, for $i\in [2^k-1]$.
Note that a homomorphism from $H$ to any subgraphs $G_i$, for $i\in [2^k-1]$, is also a homomorphism from $H$ to $G$. 
Since each $G_i$ is a subgraph of $G$, the degeneracy $\degen(G_i) \leq \degen$. Then, $\ALG$ can count homomorphisms from $H$ to any $G_i$ in time $O(mf(\degen, k))$. Using the inclusion-exclusion principle, we can obtain the number of homomorphisms from $H$ to $G$ that are also a homomorphism from $H$ to at least one of the subgraphs in $\{G_1, G_2, \ldots, G_{2^k-1}\}$ in $O(2^k \cdot mf(\degen, k))$. Hence, we can obtain the number of $\phom$s from $H$ to $G$ as follows,
\begin{align*}
   \Hom{G}{H} - \sum_{1 \leq i\leq 2^k-1} (-1)^{|F_i|-1} \Hom{G_i}{H} \,.
\end{align*}
Note that if a homomorphism from $H$ to $G$ does not include any vertex in a set $V_i$ in $\cP$, then it is also a homomorphism from $H$ to at least one of the subgraphs in $\{G_1, G_2, G_{2^k-1}\}$. Thus, we do not count such a homomorphism from $H$ to $G$. Since $k = |V(H)|$,
$\phom$s of $H$ in $G$
are actually embeddings of $H$ in $G$ that involve exactly one vertex in each part of $P$.
Observe that each such embedding of $H$ in $G$ corresponds to a $\pmatch$ of $H$ in $G$.
For each match $G_{\Hmatch}$ of $H$ in $G$, there are $|Aut(H)|$ embeddings of $H$ in $G$ that map $H$ to $G_{\Hmatch}$.
Thus, by dividing the number of $\phom$s from $H$ to $G$ by $|Aut(H)|$, we obtain the number of $\pmatch$es of $H$ in $G$ in $O(2^k \cdot mf(\degen, k))$ time.
\end{proof}

\subsection{Proof of Main Theorem}
\label{subsec:lb_proof}
We now present the details of the construction of $G_H$ for the general case and prove~\Thm{lower_bound}.
\begin{proof} [Proof of~\Cref{thm:lower_bound}]\label{proof:lower_bound_proof}
We present a linear time Turing reduction form the $\TRIC$ problem to the $\HOMC_H$ problem in bounded degeneracy graphs. Let $G$ be the input instance of the $\TRIC$ problem where $V(G)=\{u_1,\ldots,u_n\}$ and $|E(G)| = m$. First, we construct a graph $G_H$ based on $G$ and $H$ such that $G_H$ has bounded degeneracy and $O(m)$ edges.

\mypar{Construction of \bm{$G_H$}}
Let $\LICS(H)=r$ where $r\geq 6$; and $V(H)=\{a_1,a_2,\ldots,a_k\}$, where $a_{k-r+1},a_{k-r+2},\ldots,a_k,a_{k-r+1}$ is an induced $r$-cycle $C$. Let $H_{\excludeC}$ denote $H-V(C)$. $G_H$ has two main parts, $G_{\fixedGraph}$ and $G_{\coreGraph}$. These two parts are connected by the edge set $E_{\connect}$. $G_{\fixedGraph}$ is a copy of $H_{\excludeC}$ and has the vertex set $V_{\fixed}=\{z_1,z_2,\ldots,z_{k-r}\}$. $z_i \in V_{\fixed}$ corresponds to $a_i$ in $H_{\excludeC}$ for $i \in [k-r]$. Thus, $z_i,z_j \in V_{\fixed}$ are adjacent iff $\{a_i,a_j\}$ is an edge in $H_{\excludeC}$.

$G_{\coreGraph}$ contains two sets of vertices, $V_{\core}$, and $V_{\aux}$. Vertices in $V_{\core}$ correspond to vertices in $V(G)$, and vertices in $V_{\aux}$ correspond to the edges in $E(G)$.
$V_{\core}$ consists of three copies of $V(G)$ without any edges inside them. More precisely, $V_{\core}$ is composed of three sets of vertices $V_1=\{w_1,\ldots,w_n\}$, $V_2=\{x_1,\ldots,x_n\}$, and $V_3=\{y_1,\ldots,y_n\}$. For $i \in [n]$, vertices $w_i \in V_1$, $x_i \in V_2$, and $y_i \in V_3$ correspond to $u_i \in V(G)$. There are no edges inside $V_{\core}$. We describe $V_{\aux}$ next.

$V_{\aux}$ corresponds to the vertices of the paths of length $r/3$ that we add between $V_1$, $V_2$, and $V_3$. Let $r=3\ell+q$, for some $\ell \geq 2$ and $q \in \{0,1,2\}$. The vertices in $V_{\aux}$ consists of the sets of vertices $V_{1,2}$, $V_{2,3}$, and $V_{1,3}$. For each edge $e \in E(G)$ and each pair in $\{V_1,V_2,V_3\}$, we add two sets of vertices to $V_{\aux}$. Next, we describe the vertices we add to $V_{1,2}$, $V_{2,3}$, and $V_{1,3}$ for an edge $e \in E(G)$. For the pair $V_1$ and $V_2$, we add 
\begin{align*}
    V_e^{1,2} & = \left \{v_{e,1}^{1,2},\ldots,v_{e,\ell-1}^{1,2} \right \} \\
    \text{and } V_e^{2,1} & = \left \{v_{e,1}^{2,1},\ldots,v_{e,\ell-1}^{2,1} \right \}
\end{align*}
to $V_{1,2}$. For the pair $V_2$ and $V_3$, we add 
\begin{align*}
    V_e^{2,3} & = \left \{v_{e,1}^{2,3},\ldots,v_{e,\ell-1+\lfloor q/2 \rfloor}^{2,3} \right \} \\
    \text{and } V_e^{3,2} & = \left \{v_{e,1}^{3,2},\ldots,v_{e,\ell-1+\lfloor q/2 \rfloor}^{3,2} \right \}
\end{align*}
to $V_{2,3}$. And finally, for the pair $V_1$ and $V_3$, we add 
\begin{align*}
    V_e^{1,3} & = \left \{v_{e,1}^{1,3},\ldots,v_{e,\ell-1+\lfloor (q+1)/2 \rfloor}^{1,3} \right \} \\
    \text{and } V_e^{3,1} & = \left \{v_{e,1}^{3,1},\ldots,v_{e,\ell-1+\lfloor (q+1)/2 \rfloor}^{3,1}\right\}
\end{align*}
to $V_{1,3}$.
The following defines $V_{1,2}$, $V_{2,3}$, and $V_{1,3}$ more formally. For $i,j \in \{1,2,3\}$ where $i < j$,
\begin{align*}
    V_{i,j} & = \bigcup_{e \in E(G)} V_e^{i,j} \cup V_e^{j,i}.
\end{align*}
This completes the description of $V(G_{\coreGraph})$. We describe $E(G_{\coreGraph})$ next.

The edges inside $G_{\coreGraph}$ stitch vertices in $V_{\aux}$ to form paths of length $r/3$ between each pair in $\{V_1,V_2,V_3\}$. $E(G_{\core})$ consists of three sets of edges, $E_{1,2}$, $E_{2,3}$, and $E_{1,3}$. For each edge in $G$ and each pair in $\{V_1,V_2,V_3\}$, we add two sets of edges to $G_{\coreGraph}$. We describe the edges we add to $G_{\coreGraph}$ for each edge $e = \{u_i,u_j\} \in E(G)$. For the pair $V_1$ and $V_2$, we add
\begin{align*}
    E_e^{1,2} & = \left \{(w_i, v_{e,1}^{1,2}),(v_{e,1}^{1,2},v_{e,2}^{1,2}),\ldots,(v_{e,\ell-1}^{1,2}, x_j) \right \} \\
    \text{and } E_e^{2,1} & = \left \{(w_j, v_{e,1}^{2,1}),(v_{e,1}^{2,1},v_{e,2}^{2,1}),\ldots,(v_{e,\ell-1}^{2,1}, x_i) \right \}
\end{align*}
to $E_{1,2}$. Edges in $E_{1,2}$ form $\ell$-paths between $V_1$ and $V_2$ with $V_{1,2}$ as interior vertices. For the pair $V_2$ and $V_3$, we add 
\begin{align*}
    E_e^{2,3} & = \left \{(x_i, v_{e,1}^{2,3}),(v_{e,1}^{2,3},v_{e,2}^{2,3}),\ldots,(v_{e,\ell-1+\lfloor q/2 \rfloor}^{2,3}, y_j) \right \} \\
    \text{and } E_e^{3,2} & = \left \{(x_j, v_{e,1}^{3,2}),(v_{e,1}^{3,2},v_{e,2}^{3,2}),\ldots,(v_{e,\ell-1+\lfloor q/2 \rfloor}^{3,2}, y_i) \right \}
\end{align*}
to $E_{2,3}$. Edges in $E_{2,3}$ compose $\ell+\lfloor q/2 \rfloor$-paths between $V_2$ and $V_3$, by joining the vertices in $V_{2,3}$. And, for the pair $V_1$ and $V_3$, we add
\begin{align*}
    E_e^{1,3} & = \left \{(w_i, v_{e,1}^{1,3}),(v_{e,1}^{1,3},v_{e,2}^{1,3}),\ldots,(v_{e,\ell-1+\lfloor (q+1)/2 \rfloor}^{1,3}, y_j) \right \} \\
    \text{and } E_e^{3,1} & = \left \{(w_j, v_{e,1}^{3,1}),(v_{e,1}^{3,1},v_{e,2}^{3,1}),\ldots,(v_{e,\ell-1+\lfloor (q+1)/2 \rfloor}^{3,1}, y_i) \right \}
\end{align*}
to $E_{1,3}$. The edge set $E_{1,3}$ joins vertices in $V_{1,3}$ to form $\ell+\lfloor (q+1)/2 \rfloor$-paths between $V_1$ and $V_3$. We can describe the three sets of edges that compose $E(G_{\coreGraph})$ more formally as follows. For $i,j \in \{1,2,3\}$ where $i < j$,
\begin{align*}
    E_{i,j} & = \bigcup_{e \in E(G)} E_e^{i,j} \cup E_e^{j,i}.
\end{align*}

Now, we describe the edge set $E_{\connect}$ that connects $G_{\fixedGraph}$ and $G_{\coreGraph}$.
First, we partition $V_{1,2}$, $V_{2,3}$, and $V_{1,3}$ based on distance to $V_1$, $V_2$, and $V_3$, respectively. For instance, we define $V_{1,2}^i$ to be all the vertices in $V_{1,2}$ with $i$ as the length of the shortest path to a vertex in $V_1$. Recall that each vertex in $V_{1,2}$ serves as an internal vertex of a path between a vertex in $V_1$ and a vertex in $V_2$. Formally, we define
\begin{align*}
V_{1,2}^i & = \bigcup_{e \in E(G)} \{v_{e,i}^{1,2}, v_{e,i}^{2,1}\} \text{ for } i \in \{1,\ldots,\ell-1\}, \\
V_{2,3}^i & = \bigcup_{e \in E(G)} \{v_{e,i}^{2,3}, v_{e,i}^{3,2}\} \text{ for } i \in \{1,\ldots,\ell-1+\lfloor q/2 \rfloor\}, \\
\text{and } V_{1,3}^i & = \bigcup_{e \in E(G)} \{v_{e,i}^{1,3}, v_{e,i}^{3,1}\} \text{ for } i \in \{1,\ldots,\ell-1+\lfloor (q+1)/2 \rfloor\}.
\end{align*}

Now that we have partitioned $V_{\aux}$, we add the sets $V_1$, $V_2$, and $V_3$ to this partition of $V_{\aux}$ to define a partition $\mcP^\prime$ of $V(G_{\coreGraph})$ as follows.
\begin{align*}
 \mcP^\prime  & = \big\{ V_1, V_2, V_3,\\
 & V_{1,2}^1,\ldots,V_{1,2}^{\ell-1}, \\
 & V_{2,3}^1,\ldots,V_{2,3}^{\ell-1+\lfloor q/2 \rfloor}, \\
 & V_{1,3}^1,\ldots,V_{1,3}^{\ell-1+\lfloor (q+1)/2 \rfloor}\big\}.
\end{align*}
Observe that $|\mcP^\prime|=r$. Let $\map_{\connect}: V(C) \rightarrow \mcP^\prime$ be a bijective mapping. We first describe $E_{\connect}$ based on $\map_{\connect}$ and then specify $\map_{\connect}$. The following describes the edges we add to $E_{\connect}$ for each edge $e=\{u,v\} \in E(H)$ where $u\in V(C)$ and $v\notin V(C)$. Let $z_v\in V_{\fixed}$ be the vertex corresponding to $v$. We add an edge between $z_v$ and each vertex in $\map_{\connect}(u)$. We describe $\map_{\connect}$ next.

We set $\map_{\connect}$ to map $V(C)$ to an $r$-cycle in $G_{\coreGraph}$ that is a $\ppmatch$ of $C$ (recall~\Def{p-match}). Recall that $C = a_{k-r+1},a_{k-r+2},\ldots,a_k,a_{k-r+1}$. We break this cycle into three parts of length $\ell$, $\ell+\lfloor q/2 \rfloor$, and $\ell+\lfloor (q+1)/2 \rfloor$, respectively, starting from $a_{k-r+1}$. We set $\map_{\connect}(a_{k-r+1})$ to $V_1$, $\map_{\connect}(a_{k-r+1+\ell})$ to $V_2$, and $\map_{\connect}(a_{k-r+1+2\ell+\lfloor q/2 \rfloor})$ to $V_3$. In order for $\map_{\connect}$ to map $C$ to a $\ppmatch$ of $C$, we set $\map_{\connect}$ to map vertices of $C$ between $a_{k-r+1}$ and $a_{k-r+1+\ell}$ to vertices of the paths between $V_1$ and $V_2$, which are vertices in $V_{1,2}$. Similarly, $\map_{\connect}$ maps vertices of $C$ between $a_{k-r+1+\ell}$ and $a_{k-r+1+2\ell+\lfloor q/2 \rfloor}$ to $V_{2,3}$, and vertices of $C$ between $a_{k-r+1+2\ell+\lfloor q/2 \rfloor}$ and $a_{k-r+1}$ to $V_{1,3}$. Formally,
\begin{align*}
    \map_{\connect}(a_{k-r+1}) & = V_1, \\
    \map_{\connect}(a_{k-r+1+i}) & = V_{1,2}^i, \text{ for } i \in \{1,\ldots,\ell-1\}, \\
    \map_{\connect}(a_{k-r+1+\ell}) & = V_2, \\
    \map_{\connect}(a_{k-r+1+\ell+i}) & = V_{2,3}^i, \text{ for } i \in \{1,\ldots,\ell-1+\lfloor q/2 \rfloor\}, \\
    \map_{\connect}(a_{k-r+1+2\ell+\lfloor q/2 \rfloor}) & = V_3, \\
    \text{and } \map_{\connect}(a_{k-r+1+2\ell+\lfloor q/2 \rfloor+i}) & = V_{2,3}^i, \text{ for } i \in \{1,\ldots,\ell-1+\lfloor (q+1)/2 \rfloor\}.
\end{align*}
This completes the description of $E_{\connect}$ and hence $G_H$. Before presenting the details of the reduction, we first show that $G_H$ has bounded degeneracy and $O(m)$ edges.

The following lemma shows that in order to prove a graph $G$ is $t$-degenerate, we only need to exhibit an ordering $\prec$ of $V(G)$ such that each vertex of $G$ has $t$ or fewer neighbors that come later in the ordering $\prec$.
Given a graph $G$ and an ordering $\prec$ of $V(G)$, the DAG $G_\prec\dir$ is obtained by orienting the edges of $G$ with respect to $\prec$. 
\begin{lemma}\label{lem:order-and-k-degen}[Szekeres-Wilf~\cite{szekeres1968inequality}]
Given a graph $G$, $\degen(G) \leq t$ if there exists an ordering $\prec$ of $V(G)$ such that the out-degree of each vertex in $G_\prec\dir$ is at most $t$.
\end{lemma}

Next, we show that $G_H$ has bounded degeneracy using $\Lem{order-and-k-degen}$.

\begin{lemma} \label{lem:bounded-degeneracy}
$\degen(G_H) \leq k-r+2$.
\end{lemma}

\begin{proof}
We present a vertex ordering $\prec$ for $G_H$ such that for each vertex $v \in V(G_H)$, the out-degree of $v$ is at most $k-r+2$ in ${G_H}_\prec\dir$.
Let $\prec$ be an ordering of $V(G_H)$ such that $V_{\aux} \prec V_{\core} \prec V_{\fixed}$, and ordering within each set is arbitrary. Each vertex in $V_{\aux}$ is connected to exactly two other vertices in $V(G_{\coreGraph})$ and at most to all $k-r$ vertices in $V_{\fixed}$. So the out-degree of each vertex in $V_{\aux}$ in ${G_H}_\prec\dir$ is at most $k-r+2$. Since there are no edges inside $V_{\core}$, the only out-edges from vertices inside $V_{\core}$ is to vertices in $V_{\fixed}$. Further, the only out-edges from vertinces in $V_{\fixed}$ are to other vertices in $V_{\fixed}$.
Thus the out-degree of each vertex $v \in V(G_H)$ in ${G_H}_\prec\dir$ is at most $k-r+2$.
\end{proof}
Observe that $G_H$ has at most $\degen(G_H) \cdot |V(G_H)|$ edges. By \Lem{bounded-degeneracy}, $\degen(G_H) < k$, and $|V(G_H)| < 6m\ell + 3n + k$ by construction of $G_H$. Thus, $G_H$ has $O(m)$ edges.

\mypar{Details of the Reduction}
We define a partition of $V(G_H)$ by adding each vertex in $V_{\fixed}$ as a set by itself to $\mcP^\prime$. Formally,
\begin{align*}
 \mcP & = \big\{ \{z_1\},\{z_2\},\ldots,\{z_{k-r}\}, \\
 & V_1, V_2, V_3,\\
 & V_{1,2}^1,\ldots,V_{1,2}^{\ell-1}, \\
 & V_{2,3}^1,\ldots,V_{2,3}^{\ell-1+\lfloor q/2 \rfloor}, \\
 & V_{1,3}^1,\ldots,V_{1,3}^{\ell-1+\lfloor (q+1)/2 \rfloor}\big\}.
\end{align*}
Observe that $|\mcP|=k$. Also, since $G_H$ has bounded degeneracy, each subgraph of $G_H$ has bounded degeneracy too. Thus, by \Lem{hom-inc-exc} we can count $\pmatch$es of $H$ in $G_H$ in linear time if there exists an algorithm $\ALG$ for $\HOMC_H$ problem that runs in time $O(mf(\degen, k))$ for a positive function $f$. In \Lem{copy-cycle-bijection}, we prove that there is a bijection between $\pmatch$es of $H$ in $G_H$ and $\ppmatch$es of $C$ in $G_{\coreGraph}$. Further, in \Lem{triangle-linear-func}, we prove that the number of triangles in $G$ is a simple linear function of the number of $\ppmatch$es of $C$ in $G_{\coreGraph}$. So, by counting $\pmatch$es of $H$ in $G_H$, we can obtain the number of triangles in $G$.

\begin{lemma} \label{lem:copy-cycle-bijection}
There exists a bijection between the set of $\pmatch$es of $H$ in $G_H$ and the set of $\ppmatch$es of $C$ in $G_{\coreGraph}$.
\end{lemma}

\begin{proof}
Let $H^\prime$ be a $\pmatch$ of $H$ in $G_H$. Observe that by construction of $G_H$, the only edges of $G_H$ inside $V(H^\prime)$ are edges of $H^\prime$. Therefore, $H^\prime$ is actually an induced match of $H$ in $G_H$. By construction of $G_H$, specifically $E_{\connect}$, the number of edges between $V_{\fixed}$ and $V(H^\prime) \setminus V_{\fixed}$ is equal to the number of edges between $H_{\excludeC}$ and $C$. As $G_{\fixedGraph}$ is a copy of $H_{\excludeC}$, there are exactly $|E(H_{\excludeC})|$ edges inside the set of vertices $V_{\fixed}$. Thus, $H^\prime$ has exactly $|E(C)|=r$ edges inside $G_{\coreGraph}$. We describe these edges next.

Let $w_i,x_j$, and $y_t$ be the vertices of $H^\prime$ in $V_1,V_2$, and $V_3$, respectively. Inside $G_{\coreGraph}$, $w_i$ could only be connected to the two vertices of $H^\prime$ in $V_{1,2}^1$ and $V_{1,3}^1$. Furthermore, $x_j$ could only be adjacent to the two vertices of $H^\prime$ in $V_{1,2}^{\ell-1}$ and $V_{2,3}^1$. And finally, $y_t$ could only be neighbors of the two vertices of $H^\prime$ in $V_{2,3}^{\ell-1+\lfloor q/2 \rfloor}$ and $V_{1,3}^{\ell-1+\lfloor (q+1)/2 \rfloor}$.
In addition, each vertex in $V_{\aux}$ has at most two neighbors inside $G_{\coreGraph}$, and the same holds in $H^\prime$. Inside $G_{\coreGraph}$, $H^\prime$ has exactly $r$ edges, so each vertex is connected (only) to their two possible neighbors specified above. Hence, there exist an $\ell$-path between $w_i$ and $x_j$, an $\ell+\lfloor q/2 \rfloor$-path between $x_j$ and $y_t$, and an $\ell+\lfloor (q+1)/2 \rfloor$-path between $w_i$ and $y_t$. Thus, $H^\prime - V_{\fixed}$ is an $r$-cycle inside $G_{\coreGraph}$ that includes exactly one vertex in each part of $P^\prime$, and hence is a $\ppmatch$ of $C$. It is easy to see that this $\ppmatch$ of $C$ is actually an induced match. Next, we show the other direction.

Let $C^\prime$ be a $\ppmatch$ of $C$ in $G_{\coreGraph}$. It is easy to see that by construction of $G_{\coreGraph}$, $C^\prime$ is an induced match. By construction of $G_H$, $G_H[V(C^\prime) \cup V_{\fixed}]$ is an induced match of $H$. Therefore, $C^\prime$ corresponds to exactly one $\pmatch$ of $H$ in $G_H$.
\end{proof}

\begin{lemma} \label{lem:triangle-linear-func}
Let $\ppmatch(C,G_{\coreGraph})$ denote the set of $\ppmatch$es of $C$ in $G_{\coreGraph}$. The number of triangles in $G$ is equal to $|\ppmatch(C,G_{\coreGraph})|/6$.
\end{lemma}

\begin{proof}
Consider a cycle $C^\prime \in \ppmatch(C,G_{\coreGraph})$. Let $w_i$, $x_j$, and $y_t$ be the the only vertices of $C^\prime$ in $V_1$, $V_2$, and $V_3$, respectively. There should be a path between $w_i$ and $x_j$ in $C^\prime$ that does not include $y_t$, so other than $w_i$ and $x_j$, it only includes vertices in $V_{\aux}$. The only possible such path in $G_{\coreGraph}$ is an $\ell$-path between $w_i$ and $x_j$. Therefore, this $\ell$-path exists, and as a result $(u_i,u_j) \in E(G)$. Similarly, $C^\prime$ includes a path between $x_j$ and $y_t$ that only contain vertices in $V_{aux}$ other than $x_j$ and $y_t$. Therefore, there exists an $\ell+\lfloor q/2 \rfloor$-path between $x_j$ and $y_t$ in $G_{\coreGraph}$, and hence $(u_j,u_t) \in E(G)$. Finally, a path between $w_i$ and $y_t$ that other than its endpoints, only includes vertices in $V_{\aux}$, should be a part of $C^\prime$. So, there exists an $\ell+\lfloor (q+1)/2 \rfloor$-path between $w_i$ and $y_t$ in $G_{\coreGraph}$. As a result, $(u_i,u_t) \in E(G)$. Thus, $C^\prime$ corresponds only to the triangle $u_i,u_j,u_t$ in $G$. Observe that, $C^\prime$ could be specified by its vertices in $V_1,V_2$, and $V_3$.
Next, we prove the other direction; exactly 6 $\ppmatch$es of $C$ in $G_{\coreGraph}$ correspond to each triangle in $G$.

Consider a triangle $T$ with the vertex set  $\{u_i,u_j,u_t\}$ in $G$ and a $\ppmatch$ $C^\prime$ of $C$ in $G_{\coreGraph}$ that corresponds to $T$. There are six different bijective mappings from $\{u_i,u_j,u_t\}$ to $\{V_1,V_2,V_3\}$. As we showed above, $C^\prime$ could be specified by its vertices in $V_1,V_2$, and $V_3$. So, given a bijective mapping $\map_{\triMap}: \{u_i,u_j,u_t\} \rightarrow \{V_1,V_2,V_3\}$, the three vertices $\map_{\triMap}(u_i)$, $\map_{\triMap}(u_j)$, and $\map_{\triMap}(u_t)$ specify $C^\prime$. Thus, there are exactly 6 $\ppmatch$es of $C$ in $G_{\coreGraph}$ that correspond to $T$. As a result, the number of triangles in $G$ is $|\ppmatch(C,G_{\coreGraph})|/6$.
\end{proof}

\Lem{copy-cycle-bijection} and \Lem{triangle-linear-func} together show that we can obtain the number of triangles in $G$ from the number of $\pmatch$es of $H$ in $G_H$, in constant time. In conclusion, we have proved that if there exists an algorithm $\ALG$ for the $\HOMC_H$ problem that runs in time $O(mf(\degen, k))$ for a positive function $f$, then there exists an $O(m)$ algorithm for the $\TRIC$ problem. Assuming the $\TRICONJ$, the problem of $\TRIC$ has the worst case time complexity of $\omega(m)$ for an input graph with $m$ edges. Thus, the $O(m)$ Turing reduction from the $\TRIC$ problem to $\HOMC_H$ problem we presented proves \Thm{lower_bound}.
\end{proof}

\begin{observation}
\label{obs:subgraph}
In the proof of \Thm{lower_bound}, we count $\phom$s (defined in~\Def{p-match}) from $H$ to $G_H$ using the algorithm $\ALG$. Since $|\mcP|=|V(H)|$, each $\phom$ from $H$ to $G_H$ is an embedding of $H$ in $G_H$. Thus, we can apply the same argument of \Lem{hom-inc-exc} assuming there exists an algorithm for counting subgraphs, that has the same running time of $\ALG$.
Therefore, using the same argument as that of the proof of \Thm{lower_bound}, we can prove the exact same statement of $\Thm{lower_bound}$ for the $\SUBC_H$ problem.
\end{observation}

\section{Conclusion}

In this paper, we study the problem of counting homomorphisms of a fixed pattern $H$ in a graph $G$ with bounded degeneracy. We provided a clean characterization of the patterns $H$ for which near-linear time algorithms are possible ---  if and only if the largest induced cycle in $H$ has length at most $5$ (assuming standard fine-grained complexity conjectures). We conclude this exposition with two natural research directions.

While we discover a clean dichotomy for the homomorphism counting problem, the landscape for the subgraph counting problem is not as clear. Our hardness result (Theorem~\ref{thm:lower_bound}) holds for the subgraph counting version --- if a pattern $H$ has $\LICS \geq 6$, then there does not exists any near-linear time (randomized) algorithm for finding the subgraph count of $H$ (see Observation~\cref{obs:subgraph}). However, the ``only if'' direction does not follow. It would be interesting to find a tight characterization for the subgraph counting problem.

Both our current work and a previous work~\cite{bera2020linear} attempt at understanding what kind of patterns can be counted in near-linear time in sparse graphs. It would be interesting to explore beyond linear time algorithms. Specifically, we pose the following question: Can we characterize patterns that are countable in quadratic time?

\bibliographystyle{alpha}
\bibliography{refs}

\end{document}